\documentclass[aos]{imsart}\usepackage{amsthm}
\usepackage{amsmath,amsfonts,amscd,graphicx,amssymb,cite,bm,multirow,multicol,subfigure}

\arxiv{0905.3217}

\startlocaldefs
\newcommand{\p}{^{+} }
\newcommand{\m}{^{-} }
\newtheorem{thm}{Theorem}
\newtheorem{defn}{Definition}[section]
\newtheorem{cor}{Corollary}[section]
\newtheorem{prop}{Proposition}[section]
\newtheorem{propty}{Property}[section]
\newtheorem{rem}{Remark}[section]
\endlocaldefs

\begin{document}

\begin{frontmatter}

\title{Skellam Shrinkage: Wavelet-Based Intensity Estimation for Inhomogeneous Poisson Data}
\author{Keigo~Hirakawa and~Patrick~J.~Wolfe\protect\thanksref{T1}}
          \thankstext{T1}{This material is based upon work
          supported in part by the National Science Foundation under Grant DMS-0652743.  The result of Theorem~\ref{thm:skellam} was published at the IS\&T/SPIE 20th Annual Symposium on Electronic Imaging Science and Technology~\cite{HirakawaWolfe09a} in January 2009, and the IEEE International Conference on Acoustics, Speech, and Signal Processing (ICASSP)~\cite{ref:HirakawaWolfe09} in March 2009.  This theorem was also recently obtained independently by~Luisier et al.~\cite{ref:Luisier_2009}, due to appear at the 6th IEEE International Symposium on Biomedical Imaging in June 2009.
          \newline \indent The authors are with the Statistics and Information Sciences Laboratory, Harvard University, Cambridge, MA 02138 USA (e-mail: \{hirakawa, wolfe\}@stat.harvard.edu).
          }
\affiliation{Harvard University}
\runtitle{Skellam Shrinkage}
\runauthor{K.~Hirakawa and P.~J.~Wolfe}
\begin{abstract}%
The ubiquity of integrating detectors in imaging and other applications implies that a variety of real-world data are well modeled as Poisson random variables whose means are in turn proportional to an underlying vector-valued signal of interest.  In this article, we first show how the so-called Skellam distribution arises from the fact that Haar wavelet and filterbank transform coefficients corresponding to measurements of this type are distributed as sums and differences of Poisson counts.  We then provide two main theorems on Skellam shrinkage, one showing the near-optimality of shrinkage in the Bayesian setting and the other providing for unbiased risk estimation in a frequentist context.  These results serve to yield new estimators in the Haar transform domain, including an unbiased risk estimate for shrinkage of Haar-Fisz variance-stabilized data, along with accompanying low-complexity algorithms for inference.  We conclude with a simulation study demonstrating the efficacy of our Skellam shrinkage estimators both for the standard univariate wavelet test functions as well as a variety of test images taken from the image processing literature, confirming that they offer substantial performance improvements over existing alternatives.
\end{abstract}

\begin{keyword}[class=AMS]
\kwd[Primary ]{62J07}
\kwd[; secondary ]{62P30}
\kwd{62P10}
\end{keyword}

\begin{keyword}
\kwd{Haar transform}
\kwd{Poisson distribution}
\kwd{Skellam distribution}
\kwd{Unbiased risk Estimation}
\kwd{Wavelet shrinkage}
\end{keyword}

\end{frontmatter}

\section{Introduction}
\label{sec:intro}

Real-world information sensing and transmission devices are subject to various types of measurement noise; for example, losses in resolution (e.g.,~quantization effects), randomness inherent in the signal of interest (e.g.,~photon or packet arrivals), and variabilities in physical devices (e.g.,~thermal noise, electron leakage) can all contribute significantly to signal degradation. Estimation of a vector-valued signal $\bm{f}\in\mathbb{R}^N$ given noisy observations $\bm{g}\in\mathbb{R}^N$ therefore plays a prominent role in a variety of engineering applications such as signal processing, digital communications, and imaging.

At the same time, statistical modeling of transform coefficients as latent variables has enjoyed tremendous popularity across these diverse applications---in particular, wavelets and other filterbank transforms provide convenient platforms; as is by now universally acknowledged, such classes of transform coefficients tend to exhibit temporal and spectral decorrelation and energy compaction properties for a variety of data. In this setting, the special case of additive white Gaussian noise is by far the most studied scenario, as the posterior distribution of coefficients is readily accessible when the likelihood function admits a closed form in the transform domain.

The twin assumptions of additivity and Gaussianity, however, are clearly inadequate for many genuine engineering applications; for instance, measurement noise is often dependent on the range space of the signal $\bm{f}$, effects of which permeate across multiple transform coefficients and subbands~\cite{ref:Hirakawa_2008c}.   For instance, the number of photoelectrons $g_i$ accumulated by the $i$th element of a photodiode sensor array---an integrating detector that ``counts photons''---is well modeled as a Poisson random variable $g_i \sim \mathcal{P}(f_i)$, where $f_i$ is proportional to the average incident photon flux density at the $i$th sensor element.

Recall that for $g_i \sim \mathcal{P}(f_i)$ we have that $\operatorname{\mathbb{E}} g_i = \operatorname{Var}g_i = f_i$, and so in the case at hand $f_i$ reflects (up to quantum efficiency) the $i$th expected photoelectron count, with the resultant ``noise'' in the form of variability being signal-dependent and hence heteroscedastic.  Indeed, the local signal-to-noise ratio at the $i$th sensor element is seen to grow linearly with signal strength as $\operatorname{\mathbb{E}} g_i^2 / \operatorname{Var}g_i $  = $ 1 + f_i$, implying very noisy conditions when dealing with inefficient detectors or low photon counts.

Classical variance stabilization techniques dating back to Bartlett and Anscombe~\cite{ref:Bartlett_1936, ref:Anscombe_1948, ref:Freeman_1950, ref:Fisz_1955, ref:Tweedie_1968, ref:Veevers_1971} yield an approach to Poisson mean estimation designed to recover homoscedasticity, with~\cite{ref:Fryzlewicz_2004} providing a summary of more recent work.   Here one seeks an invertible operator $\bm{\gamma}:\mathbb{Z}_+^N \to \mathbb{R}^N$, typically by way of a compressive nonlinearity such as the component-wise square root, that (approximately) maps the heteroscedastic realizations of an inhomogeneous Poisson process to the familiar additive white Gaussian setting:
\begin{align*}
g_i \sim \mathcal{P}(f_i), \, i\in\{1,2,\ldots,N\} \quad \mapsto \quad
\bm{\gamma}(\bm{g}) \sim \mathcal{N}(\bm{\gamma}(\bm{f}),\bm{I}_N)\text{.}
\end{align*}
Standard techniques may then be used to estimate $\bm{\gamma}(\bm{f})$ directly, with the inverse transform
$\bm{\gamma}^{-1}(\cdot)$ applied post hoc.

Inhomogeneous Poisson data can also be treated directly.  For instance, empirical Bayes approaches leverage the independence of Poisson variates via their empirical marginal distributions~\cite{ref:Robbins_1956,ref:Raphan_2007}, while multiparameter estimators borrow strength to improve upon maximum-likelihood estimation~\cite{ref:Clevenson_1975, ref:Hudson_1978, ref:Ghosh_1983}; however, this ignores potential correlations amongst elements of $\bm{f}$.  To address such concerns, multiresolution approaches to Poisson intensity estimation were introduced to explicitly encode the dependencies between the Poisson variables in the context of Haar frames~\cite{ref:Timmermann_1999, ref:Kolaczyk_1999b, ref:Nowak_2000, ref:Kolaczyk_2004}.  The relative merits of the various methods described above are well documented~\cite{ref:Tweedie_1968, ref:Veevers_1971, ref:Besbeas_2004, ref:Jansen_2006, ref:Willett_2007} and will not be repeated here.

In this paper, we address Poisson rate estimation directly in the Haar wavelet and Haar
filterbank transform domains by way of the \emph{Skellam distribution}~\cite{ref:Skellam_1946}, whose use to date has been limited to special settings~\cite{ref:Karlis_2003, ref:Karlis_2006, ref:Hwang_2007a, ref:Hwang_2007b, ref:Zhang_2008}.  After briefly reviewing wavelet and filterbank coefficient models in Section~\ref{sec:wavelets}, we then describe in Section~\ref{sec:estimation} new Bayesian and frequentist transform-domain estimators for both exact and approximate inference.  Here we first derive posterior means under canonical heavy-tailed priors, along with analytical approximations to the optimal estimators that we show to be both efficient and practical.  We then show how inhomogeneous Poisson variability leads to a variant of Stein's unbiased risk estimation~\cite{ref:Stein_1981} for parametric estimators in the transform domain.  Simulation studies presented in Section~\ref{sec:results} verify the effectiveness of our approach, and we conclude with a brief discussion in Section~\ref{sec:discussion}.

\section{Wavelet and Filterbank Coefficient Models}
\label{sec:wavelets}

\subsection{Haar Wavelet and Filterbank Transforms}
\label{sec:review_wavelets}

Consider a nested sequence of closed subspaces $\{V_k\}_{k\in Z}$ of $L^2(\mathbb{R})$ satisfying the axioms required of a multiresolution analysis~\cite{ref:Mallat_1998}.  Then there exists a \emph{scaling function} $\phi \in L^2(\mathbb{R})$ such that the family $\{2^{-k/2}\phi\left(2^{-k}(\cdot-i)\right)\}_{i\in\mathbb{Z}}$ is an orthonormal basis of $V_k$ for all $k\in\mathbb{Z}$.   There also exist a corresponding \emph{conjugate mirror filter} sequence $\{h_i\}_{i\in\mathbb{Z}}$ and \emph{admissible wavelet} $\psi$, with Fourier transforms $\hat{h},\hat{\psi}$ respectively, satisfying
\begin{equation*} % Mallat 7.31, 7.61
\begin{cases}
\hat{\phi}(2\omega)=& 2^{-1/2} \hat{h}(\omega)\hat{\phi}(\omega) \text{,}\\
\hat{\psi}(2\omega)=& 2^{-1/2} e^{-j\omega} \hat{h}^*(\omega-\pi)\hat{\phi}(\omega) \text{.}
\end{cases}
\end{equation*}
Moreover, for any fixed scale $2^k$ the wavelet family $\{2^{-k/2}\psi\left(\cdot/2^k -i\right)\}_{i\in\mathbb{Z}}$ forms an orthonormal basis of the orthogonal complement of $V_k$ in $V_{k-1}$, and for all $(i,k)\in\mathbb{Z}^2$ the wavelet families together comprise an orthobasis of $L^2(\mathbb{R})$.

Recursively expanding the above $K$ times, and defining
\begin{align*}
\begin{cases}
\text{\emph{Wavelet coefficient}} & x_{k,i} := \langle f,2^{-k/2}\psi(\cdot/2^k-i)\rangle \\
\text{\emph{Scaling coefficient}} & s_{k,i} := \langle f,2^{-k/2}\phi(\cdot/2^k-i)\rangle \text{,}
\end{cases}
\end{align*}
we see that any $f \in L^2(\mathbb{R})$ admits the following orthobasis expansion in terms of its wavelet and scaling coefficients:
\begin{align*}
f&=\sum_{i=-\infty}^{\infty}s_{0,i} \, \phi(\cdot-i)\\
&=\sum_{i=-\infty}^{\infty}\frac{s_{K,i}}{2^{\frac{K}{2}}}\phi\left(\frac{\cdot-2^Ki}{2^K}\right)
+\sum_{k=1}^{K}\sum_{i=-\infty}^{\infty}\frac{x_{k,i}}{2^{\frac{k}{2}}}\psi\left(\frac{\cdot-2^ki}{2^k}\right)\text{.}
\end{align*}
The mapping
$f\mapsto\{s_{K,i},x_{k,i}\}$ is termed a $K$-level continuous wavelet transform, with an analogous discrete wavelet transform defined for sequences in $\ell^2(\mathbb{Z})$.

For the special case of a Haar wavelet transform, we take as our scaling function $\phi = \mathbb{I}_{[0,1]}$
(the unit indicator), with $h_i = \langle 2^{-1/2} \phi(\cdot/2),\phi(\cdot-i) \rangle$ yielding $h_0=h_1=2^{-1/2}$ as the only nonzero conjugate mirror filter values.  This in turn induces a recursive relationship as follows:
\begin{equation}
\label{eqn:HWT}
\begin{cases}
x_{k,i}=s_{k-1,2i}-s_{k-1,2i+1} \text{,} \\
s_{k,i}=s_{k-1,2i}+s_{k-1,2i+1} \text{.}
\end{cases}
\end{equation}

In fact, this one-level transform is a version of a filterbank transform---a
canonical multirate system of the type used for time-frequency analysis in
digital signal processing. That is, $\hat{h}$ satisfies the perfect reconstruction
condition~\cite{ref:Mallat_1998}
\begin{equation*}
\begin{cases}
\hat{h}^*(\omega)\hat{h}(\omega)+\hat{h}^*(\omega-\pi)\hat{h}(\omega-\pi)=\text{const} \text{,} \\
\hat{h}^*(\omega)\hat{h}(\omega-\pi)+\hat{h}^*(\omega-\pi)\hat{h}(\omega)=0 \text{.}
\end{cases}
\end{equation*}

In the formulation of~\eqref{eqn:HWT}, each sequence $\{s_{k-1,i}\}_i$ is decomposed into lowpass and highpass components $\{s_{k,i},x_{k,i}\}_i$ in turn. A recursive application of the map
$\{s_{k-1,i}\}\mapsto\{s_{k,i},x_{k,i}\}$ yields the Haar wavelet
transform, whereas the same transform applied to highpass component $x_{k-1,i}$ further
decomposes it into narrower bands. Recursive decomposition of
\emph{both} lowpass and highpass sequences in this way yields the Hadamard
transform, otherwise known as the Haar filterbank transform.

The low computational requirements of these transforms make them attractive alternatives to other joint
time-frequency analysis techniques possessing better frequency
localization. The Haar transforms enjoy orthogonality, compact spatial support, and
computational simplicity, with the Haar wavelet transform satisfying the axioms of a multiresolution
analysis. We later demonstrate how their simplicity serves to admit analytical
tractability that in turn enables efficient inference and estimation procedures.

As a final note, we omit subband index $k$ in the sequel, as wavelet coefficients
$x_{k,i}$ are always aggregated within a given scale $2^k$; for notational clarity in the finite-dimensional setting, further suppression of subscript $i$ will be used to indicate a generic scalar coefficient $x_{(\cdot)}$, as distinct from vector-valued quantities (e.g., $\bm{x}$) indicated in bold throughout.

\subsection{Transform-Domain Denoising}
\label{sec:review_denoising}

Turning to the problem of transform-domain denoising, consider the case whereupon a vector
of noisy orthobasis coefficients ${\bm{y}}\sim\mathcal{N}(\bm{x},{\sigma}^2\bm{I}_N)$ is observed, with $\bm{x}$ deterministic but unknown. Writing an estimator for $\bm{x}$ as $\widehat{\bm{X}}(\bm{Y}) = \bm{Y} + \bm{\theta}(\bm{Y})$, Stein's Lemma~\cite{ref:Stein_1981} may be used to formulate an unbiased estimate of the associated $\ell^2$ risk $\operatorname{\mathbb{E}} \|\widehat{\bm{X}} - \bm{x}\|_2^2$ as follows.

\begin{thm}[Stein's Unbiased Risk Estimate (SURE)~\cite{ref:Stein_1981}]\label{lem:stein_lem}
Let $\bm{y} \sim \mathcal{N}(\bm{x},\sigma^2\bm{I}_N)$, with $\bm{x}$ unknown, and fix an estimator $\widehat{\bm{X}}(\bm{Y})=\bm{Y} + \bm{\theta}(\bm{Y})$ such that $\bm{\theta}:\mathbb{R}^N\to\mathbb{R}^N$ is weakly differentiable.  Then the resultant risk may be formulated as
\begin{equation}
\label{eqn:Stein}
\operatorname{\mathbb{E}} \| \widehat{\bm{X}}\!(\bm{Y})-\bm{x} \|_2^2
=  N{\sigma}^2 + \operatorname{\mathbb{E}} \left[ \| \bm{\theta}(\bm{Y}) \|_2^2 + 2{\sigma}^2 \operatorname{div} \bm{\theta}(\bm{Y}) \right] \text{,}
\end{equation}
with $N{\sigma}^2 + \| \bm{\theta}(\bm{y}) \|_2^2 + 2{\sigma}^2 \operatorname{div}\bm{\theta}(\bm{y})$ an unbiased estimate thereof.
\end{thm}

Hence, by replacing the latter expectation of~\eqref{eqn:Stein} with an evaluation over the vector $\bm{y}$ of observed transform coefficients, one may directly optimize parameter choices for nonlinear shrinkage estimators---for example soft thresholding, given by
\begin{equation}
\label{eqn:softThresh}
\hat{X}_i({Y}_i; \tau) := \operatorname{sgn}({Y}_i)\max(|{Y}_i|-\tau,0) \text{.}
\end{equation}
As an example that we shall return to later, SUREShrink~\cite{ref:Donoho_1995} is obtained from~\eqref{eqn:Stein} and~\eqref{eqn:softThresh} by writing $\widehat{\bm{X}}(\bm{Y}) = \bm{Y} + \bm{\theta}(\bm{Y}; \tau)$:
\begin{align}\label{eqn:st_theta}
\theta({Y}_i; \tau)=&\begin{cases}
-\operatorname{sgn}({Y}_i)\,\tau&\text{if $|{Y}_i|\geq\tau$}\\
-{Y}_i&\text{if $|{Y}_i|<\tau$}
\end{cases}\\
\frac{\partial}{\partial{Y}_i}\theta({Y}_i; \tau)=&\begin{cases}
0&\text{if $|{Y}_i|\geq\tau$}\\
-1&\text{if $|{Y}_i|<\tau$} \text{,}\notag
\end{cases}
\end{align}
and thus $\tau$ is chosen to minimize the empirical risk estimate
\begin{equation}
\label{eqn:SURE}
N\sigma^2 + \sum_{i=1}^N \min (y_i^2,\tau^2) - 2\sigma^2\,\#\{i: |y_i|<\tau\}
\text{.}
\end{equation}

\subsection{The Skellam Distribution}
\label{sec:skellam}

In contrast to the above setting of additive white Gaussian noise, the distribution of inhomogeneous Poisson data $\bm{g}: g_i \sim \mathcal{P}(f_i)$ is \emph{not} invariant under orthogonal transformation---and so transform-domain denoising ceases to be as straightforward in the general setting~\cite{ref:Hirakawa_2008c}.  However, for the special cases of the Haar wavelet and filterbank transforms described in Section~\ref{sec:review_wavelets}, we may characterize their coefficient distributions in closed form as sums and differences of Poisson counts.

To this end, let the matrix $\bm{W}\in\{0,\pm 1\}^{N\times N}$ denote an (unnormalized) Haar filterbank transform.  Taking $\bm{x} := \bm{W} \! \bm{f}$ to be the transform of $\bm{f} \in \mathbb{Z}_+^N$, the resultant wavelet and scaling coefficients comprise sums and differences of elements of $\bm{f}$:
\begin{align}
\bm{x} := \bm{W} \! \bm{f} & \,\Rightarrow\,
\begin{cases}
\text{\emph{Wavelet coefficient}} & x_i = x_i\p - x_i\m \text{;}\\
\text{\emph{Scaling coefficient}} & s_i = x_i\p + x_i\m \text{;}
\end{cases} \label{eqn:wavsca} \\
x_i\p & := \sum_{j:W_{ij}=1}f_j, \quad x_i\m : = \sum_{j:W_{ij}=-1}f_j\text{.}  \label{eqn:pm}
\end{align}

An analogous definition with respect to the observed data
$g_i \sim \mathcal{P}(f_i)$ and its Haar filterbank transform $\bm{y} := \bm{W} \! \bm{g}$ implies that
the empirical wavelet and scaling coefficients themselves comprise sums and differences of Poisson counts:
\begin{align}
\bm{W} \! \bm{g} & \Rightarrow
\begin{cases}
\text{\emph{Empirical wavelet coefficient}} & \!\!\!\! y_i = y_i\p - y_i\m \text{;} \\
\text{\emph{Empirical scaling coefficient}} & \!\!\!\! t_i = y_i\p + y_i\m \text{;}
\end{cases} \label{eqn:ObsWav} \\
& y_i\p \sim \mathcal{P}(x_i\p), \quad y_i\m \sim \mathcal{P}(x_i\m),
\quad t_i \sim \mathcal{P}(s_i) \text{.} \label{eqn:ypm}
\end{align}

Thus the empirical coefficients defined by~\eqref{eqn:ObsWav} are effectively corrupted versions of those in~\eqref{eqn:wavsca}.  While the sum of Poisson variates $y_i\p$ and $y_i\m$ is again Poisson, as indicated by the expression of~\eqref{eqn:ypm} for empirical scaling coefficient $t_i$, the distribution of their difference also admits a closed-form expression, first characterized by Skellam~\cite{ref:Skellam_1946} using generating functions.

\begin{prop}
\label{prop:PoissonDiff}
Fix $x\p,x\m \in \mathbb{R}_+$, and let the random variable $Y\in\mathbb{Z}$ denote the difference of two Poisson variates $y\p \sim \mathcal{P}(x\p)$ and $y\m \sim \mathcal{P}(x\m)$.  Defining $I_{y}(\cdot)$ to be the $y$th-order modified Bessel function of the first kind, we have that
\begin{align}\label{eqn:SkellamPoisson}
\operatorname{Pr}(Y\!=\!y\,;x\p,x\m) \!= & \, e^{-(x\p\!+x\m)}\!\left(\!\frac{x\p}{x\m}\!\right)^{\!\frac{y}{2}}\!I_{y\!}\!\left(2\sqrt{x\p x\m}\right) \text{,}
\\ & y \in \mathbb{Z}; \quad x\p,x\m \in \mathbb{R}_+ \text{.} \notag
\end{align}
\end{prop}

\begin{proof}
A direct verification is provided by series representations of Bessel functions~\cite{ref:Gradshteyn_2007}.  First, note that via correlation of Poisson densities we obtain directly
\begin{equation}\label{eq:skellamConv}
\operatorname{Pr}(Y\!=\!y\,;x\p,x\m) = e^{-(x\p+x\m)} \!\!\!\!\!\!\!\! \sum_{k=\max(y,0)}^\infty \!\!\!\!\! \frac{(x\p)^{k}(x\m)^{k-y}}{k!\,(k-y)!}
\text{,}
\end{equation}
By change of variables in the summation index of~\eqref{eq:skellamConv} according to $\max(y,0) = (|y|+y)/2$, we obtain a summand that is symmetric in $y \in \mathbb{Z}$ as follows:
\begin{equation*}
\operatorname{Pr}(Y\!=\!y\,;x\p,x\m) = e^{-(x\p+x\m)}
\left(\!\frac{x\p}{x\m}\!\right)^{\!\frac{y}{2}} \sum_{k=0}^\infty \frac{(x\p x\m)^{k+\frac{|y|}{2}}}{k!\,(|y|+k)!}
\text{.}
\end{equation*}
The result follows from the observation that $I_{\nu}(\cdot)$ admits, for positive argument and order, the real-valued Taylor expansion
\begin{equation*}
I_{\nu\!}(t) = \sum_{k=0}^\infty \frac{\left(t/2\right)^{2k+\nu}}{k! \, \Gamma(\nu+k+1)}; \quad \nu,t \in \mathbb{R}_+ \text{,}
\end{equation*}
coupled with the fact that $I_{-\nu}(\cdot) = I_{\nu}(\cdot)$ for $\nu \in \mathbb{N}$. %integral $\nu$.
\end{proof}

We have thus proved that the distribution of each empirical coefficient $y_i = y\p_i - y\m_i$ in~\eqref{eqn:ObsWav} %can be described via the \emph{Skellam distribution}~\cite{ref:Skellam_1946} as follows.
may be described as follows.
\begin{defn}[Skellam Distribution~\cite{ref:Skellam_1946}]
\label{def:skellam}
Let $Y \in \mathbb{Z}$ denote a difference of Poisson variates according to~\eqref{eqn:wavsca}--\eqref{eqn:ypm}, with index $i$ suppressed for clarity as in Proposition~\ref{prop:PoissonDiff}.  Then
\begin{equation*}
\operatorname{\mathbb{E}} Y = x\p - x\m = x, \quad \operatorname{Var} Y = x\p + x\m = s \text{,}
\end{equation*}
where $s \geq |x|$, and variate $y$ takes the \emph{Skellam distribution}:
\begin{align}
& y \sim \mathcal{S}(x,s); \qquad s \in \mathbb{R}_+, \, -s \leq x \leq s \notag \\
& %\operatorname{Pr}(Y\!\!=\!y;x,s)
p(y\,;x,s)
= e^{-s}\left(\!\frac{s+x}{s-x}\!\right)^{\!\frac{y}{2}}\!I_{y\!}\!\left(\sqrt{s^2-x^2}\right)
\label{eq:skellamPMF}
\text{.}
\end{align}
\end{defn}

\begin{rem}[Support and Limiting Cases]
\end{rem}\vspace{-0.75\baselineskip}
As the difference of two Poisson variates, a Skellam variate ranges over the integers unless either $x\p, x\m=0$, in which case a direct appeal to the discrete convolution of~\eqref{eq:skellamConv} recovers the limiting Poisson cases.
On the other hand, as both $x\p, x\m \to \infty$, it follows from the Central Limit Theorem that the distribution of a Skellam variate tends toward that of a Normal. % $\mathcal{N}(x,s)$.

\begin{rem}[Skewness and Symmetry]
\end{rem}\vspace{-0.75\baselineskip}
The skewness of a Skellam random variable is easily obtained from its generating function as $s^{-3/2}x$~\cite{ref:Skellam_1946}, and hence is proportional to the difference in Poisson means $x\p$ and $x\m$, with a rate that grows in inverse proportion to their sum.  Indeed, when $x=0$ the distribution is symmetric,
with variance $s$ proportional to the geometric mean of $x\p$ and $x\m$ according to~\eqref{eqn:SkellamPoisson}.  A standard $\mathcal{S}(0,1)$ Skellam random variable is shown in Fig.~\ref{fig:SkellamStandard},
\begin{figure*}[!t]
\centering
\subfigure[\label{fig:SkellamStandard}Standard Skellam distribution $\mathcal{S}(0,1)$]{%
\includegraphics[width=.49\textwidth]{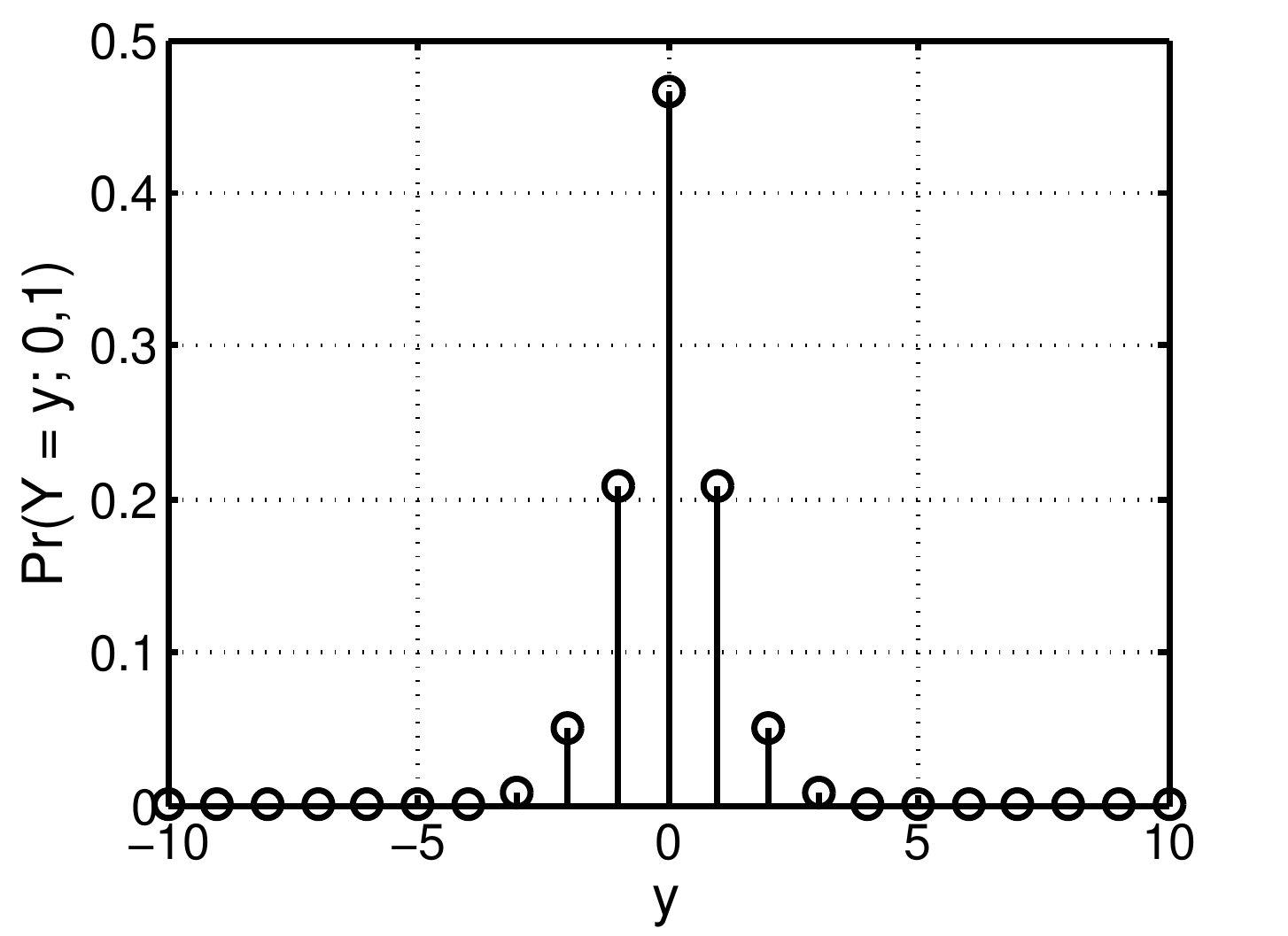}}\\
\subfigure[\label{fig:SkellamTail}Tail behavior with increasing variance]{%
\includegraphics[width=.49\textwidth]{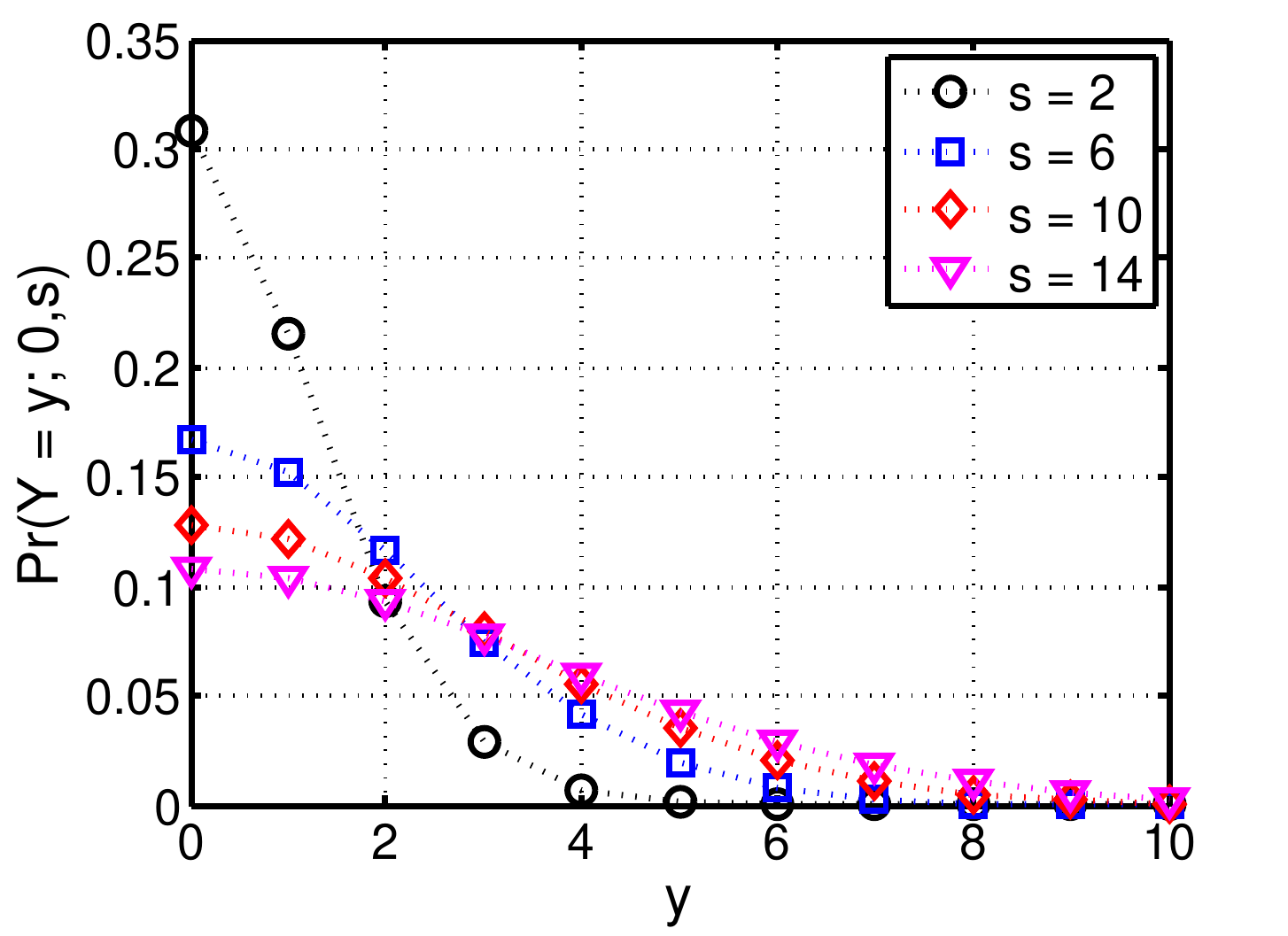}}
\subfigure[\label{fig:SkellamSkew}Skewness in terms of mean and variance]{%
\includegraphics[width=.49\textwidth]{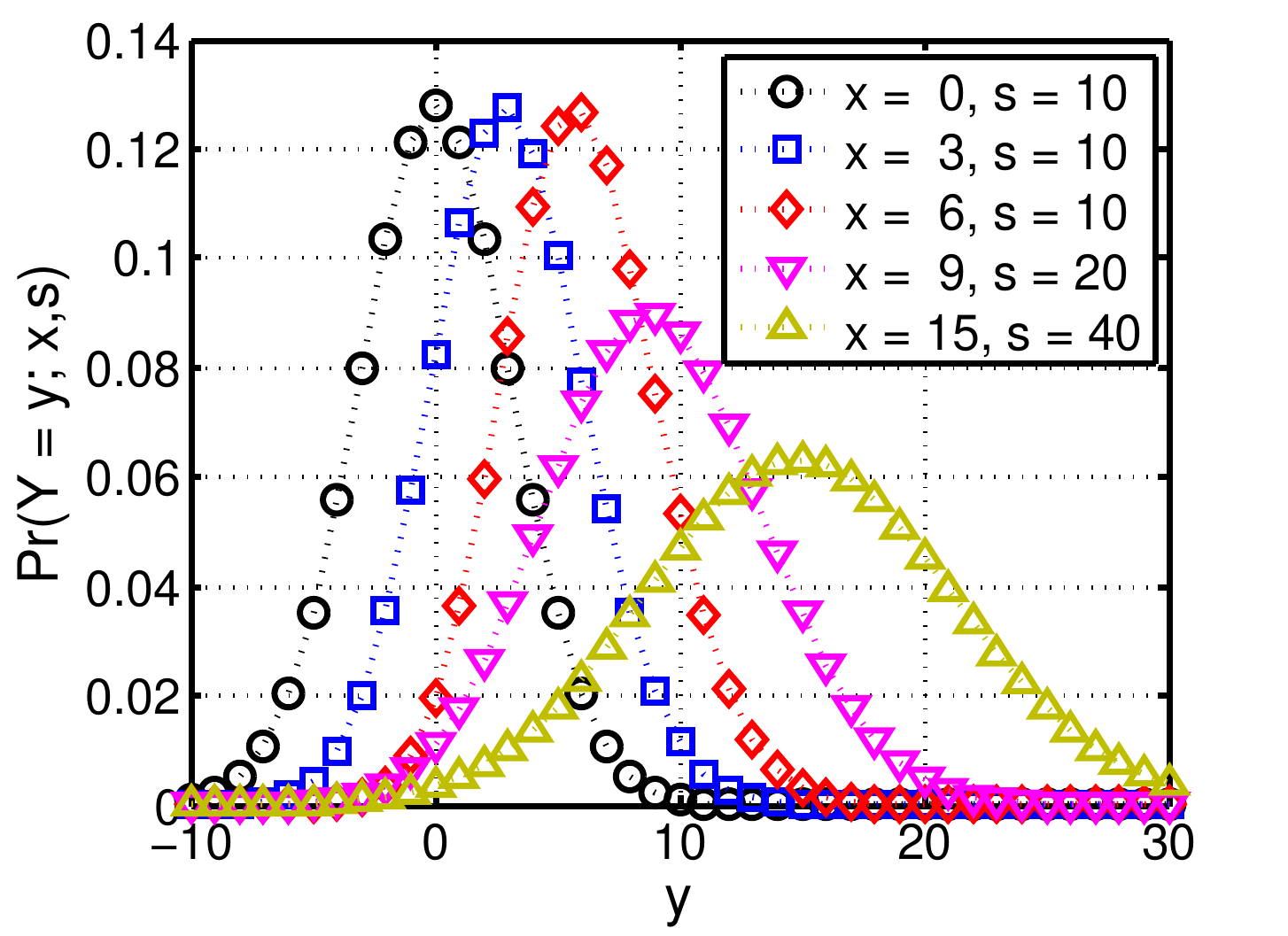}}
\caption{Illustrations of the Skellam distribution $\mathcal{S}(x,s)$ showing tail behavior and skewness.  See Definition~\ref{def:skellam} in the text for details}
\label{fig:Skellam}
\end{figure*}
with Fig.~\ref{fig:SkellamTail} detailing the tail behavior of other symmetric cases $\mathcal{S}(0,s)$; examples illustrating skewness as a general function of mean and variance are shown in Fig.~\ref{fig:SkellamSkew}.

Returning now to our context of Haar transforms, we next observe that the density of empirical coefficient $y_i$ depends only on the corresponding wavelet and scaling coefficients $x_i$ and $s_i$ (and similarly for the coarsest Haar wavelet subband).
\begin{prop}\label{prop:condInd}
Let $y_i \sim \mathcal{S}(x_i,s_i)$ according to Definition~\ref{def:skellam}, with $\bm{x} := \bm{W} \! \bm{f}$ a vector of Haar filterbank transform coefficients, and $\bm{y}$ that of the empirical coefficients.  Then
\begin{align*}
p(y_i\,;\bm{f})=p(y_i \,; s_i,x_i)\text{.}
\end{align*}
\end{prop}

\begin{proof}
The relation is a straightforward consequence of the choice of transform. From the definitions in~\eqref{eqn:pm},
\begin{align*}
p(y_i\,;\bm{f}) & = p(y_i\,;x_i\p,x_i\m)\\
& = p\!\left(y_i \,; \textstyle \sum_{j:W_{ij}=1}f_j,\sum_{j:W_{ij}=-1}f_j \right)\\
& = p\!\left(y_i \,; \textstyle \sum_{j:W_{ij}=1}(\bm{W}^{-1}\bm{x})_{j}, \sum_{j:W_{ij}=-1}(\bm{W}^{-1}\bm{x})_j\!\right)\!\text{.}
\end{align*}
Let $\bm{v}_i$ and $\bm{w}_i$ be row vectors from $\bm{W}\!$ such that $s_i=\bm{v}_i\bm{f}$ and $x_i=\bm{w}_i\bm{f}$, respectively.  It is easily verified that the $j$th entry of $(\bm{v}_i + \bm{w}_i)/2$ is nonzero if and only if $W_{ij}=1$, and hence
\begin{align*}
p(y_i\,;f)
& = p\!\left(y_i \,; \left(\textstyle{\frac{\bm{v}_i+\bm{w}_i}{2}}\right)\bm{W}^{-1}\bm{x}, \left(\textstyle{\frac{\bm{v}_i-\bm{w}_i}{2}}\right)\bm{W}^{-1}\bm{x} \right) \\
&=p\!\left(y_i \,; \textstyle{\frac{s_i+x_i}{2}},\textstyle{\frac{s_i-x_i}{2}} \right) =p(y_i \,; s_i,x_i)\text{.}
\end{align*}
\end{proof}

\section{Wavelet-Domain Poisson Intensity Estimation}
\label{sec:estimation}

Recall our goal of leveraging properties of Haar wavelets and filterbanks to accomplish transform-domain intensity estimation for inhomogeneous Poisson data.  To this end there are two main conclusions to be drawn from Section~\ref{sec:skellam} above:  First, Poisson variability in the data domain gives rise to Skellam variability in Haar transform domains (Definition~\ref{def:skellam}).  Second, the conditional independence structure of Haar coefficients suggests univariate Skellam estimators as a first step toward achieving satisfactory performance (Proposition~\ref{prop:condInd}).

Accordingly, we now turn our attention to deriving univariate Skellam mean estimators under both Bayes and frequentist assumptions.  We work throughout with the generic scalar quantity $Y \sim \mathcal{S}(x,s)$, where Haar scaling coefficient $s$ is given and Haar wavelet coefficient $x$ is a latent variable, assumed to be random or deterministic depending on context.  Although the scaling coefficient is not directly observed in practice, this standard wavelet estimation assumption amounts to using the empirical scaling coefficient $t_i$ of~\eqref{eqn:ObsWav} as a plug-in estimator of $s_i$ in~\eqref{eqn:wavsca}.  As Haar scaling coefficients constitute sums of Poisson variates in this context, their expected signal-to-noise ratios are likely to be high, in keeping with the arguments of Section~\ref{sec:intro}, and moreover they admit asymptotic Normality.

\subsection{Key Properties of the Skellam Likelihood Model}
\label{sec:SkellamLikelihoodProps}

We first develop some needed properties of the Skellam likelihood model;
while these follow from standard recurrence relations for Bessel functions of integral order, probabilistic derivations can prove more illuminating.  We begin with expressions for partial derivatives of the Skellam distribution.

\begin{propty}[Derivatives of the Skellam Likelihood]\label{prop:SkellamPartials}
Partial derivatives of the Skellam likelihood $p(y\,;x,s)$
admit the following finite-difference expressions:
\begin{align*}
\frac{\partial}{\partial x} p(y\,;x,s) & = \frac{1}{2}\left[p(y\!-\!1\,;x,s)-p(y\!+\!1\,;x,s)\right] \\
\frac{\partial}{\partial s} p(y\,;x,s) & = \frac{1}{2}\left[p(y\!-\!1\,;x,s)+p(y\!+\!1\,;x,s)\right] - p(y\,;x,s)
\text{.}
\end{align*}
\end{propty}

\begin{proof}
Recall from Definition~\ref{def:skellam} that a Skellam variate $Y \sim \mathcal{S}(x,s)$ comprises the difference of two Poisson variates with respective means $x\p$ and $x\m$.  Denoting by $\mathcal{F}$ the (conjugate) Fourier transform operator acting on the corresponding probability measure, its characteristic function in $\omega$ follows as
\begin{align*}
\mathcal{F}p(y\,;x,s) & = \exp\left[ x\p (e^{j\omega} - 1) + x\m (e^{-j\omega} - 1) \right] \\
& = \textstyle \exp\left[ \frac{1}{2}(s+x)  e^{j\omega} + \frac{1}{2}(s-x) e^{-j\omega} - s \right]
\text{;}
\end{align*}
and hence, invoking linearity, we may compute derivatives as:
\begin{align*}
\frac{\partial}{\partial x}p(y\,;x,s) & = \mathcal{F}^{-1}\frac{\partial}{\partial x}(\mathcal{F}p)(\omega) =\mathcal{F}^{-1}{\textstyle\frac{1}{2}}\left(e^{j\omega}\!-\!e^{-j\omega}\right) (\mathcal{F}p)
\text{,}
\end{align*}
and similarly for the partial derivative of $p(y\,;x,s)$ in $s$. % Skellam partial derivative in $s$.
\end{proof}

\begin{rem}
\end{rem}\vspace{-0.75\baselineskip}
Property~\ref{prop:SkellamPartials} implies that $(\partial / \partial x) p(y\,;x,s)$ is the normalized first central difference of the likelihood on its domain $y$, and that $(\partial / \partial s) p(y\,;x,s)$ is one-half the normalized second central difference.  Hence slope and curvature of the likelihood are encoded directly in the Skellam score functions.

Next, we note that for $\nu \in \mathbb{N}$ the standard Bessel identity $I_{\nu}(t) = - 2(\nu-1)/t \, I_{\nu-1}(t) +I_{\nu-2}(t)$ implies the following.

\begin{propty}[Skellam Likelihood Recursion]
\label{prop:SkellamRecursion}
The Skellam likelihood $p(y\,;x,s)$ admits the following recurrence relation in $y$ for fixed $(x,s)$:
\begin{equation*}
p(y\,;x,s) = \frac{-2(y-1)}{s-x} \, p(y-1\,;x,s) + \frac{s+x}{s-x} \, p(y-2\,;x,s)
\text{.}
\end{equation*}
\end{propty}

\begin{rem}
\end{rem}\vspace{-0.75\baselineskip}
This property lends itself to easy calculation of the Skellam likelihood, as fixed initial values may be tabulated and used to initialize the recursion, thus avoiding the evaluation of Bessel functions.

Combining Properties~\ref{prop:SkellamPartials} and~\ref{prop:SkellamRecursion}, %implies the following.
we have our final result.

\begin{propty}[Skellam Differential Equation]\label{prop:deriv}
The Skellam likelihood $p(y\,;x,s)$
satisfies a linear, first-order hyperbolic partial differential equation in $(x,s)$, for fixed $y$, as follows:
\begin{equation}
(y-x) \, p(y\,;x,s) = s\, \frac{\partial}{\partial x} p(y\,;x,s) + x\, \frac{\partial}{\partial s} p(y\,;x,s) \text{.} \label{eqn:SkellamPDE}
\end{equation}
\end{propty}

\subsection{Prior Models and Posterior Inference via Shrinkage}
\label{sec:BayesEst}

Having developed needed properties of the Skellam likelihood $p(y\,;x,s)$ above, and with $s$ assumed directly observed, we now consider the setting in which each underlying transform coefficient $x: |x| \leq s$ is modeled as a random variable.  While determining the most appropriate choice of prior distribution for different problem domains remains an open area of research, with examples ranging from generalized Gaussian distributions through discrete and continuous scale mixtures, we make no attempt here to introduce new insights on prior elicitation.  Rather, we focus on optimal estimation for general classes of prior distributions having compact support.

The problem being univariate, exact inference is realizable through numerical methods; however, the requisite determination of prior parameters, possibly from data via empirical Bayes, renders this approach infeasible in practice, as posterior values cannot be easily tabulated in advance.  To this end, the main result of this section is an approximate Skellam conditional mean estimator with bounded error, obtained as a closed-form shrinkage rule.
\begin{thm}[Skellam Shrinkage]\label{thm:SkellamPost}
Consider a Skellam random variable $Y \sim \mathcal{S}(X,s)$, with $s$ fixed but $X$ a random variable that admits a density with respect to Lebesgue measure on $[-s,s]$.  Define the Bayes point estimator
\begin{equation}\label{eqn:SkellamPost}
\widehat{X} := Y - s \textstyle
\operatorname{\mathbb{E}}\left(\frac{\partial}{\partial X} \ln p(Y\,\vert\,X;s) \,\big\vert\,Y;s\right) \text{,}
\end{equation}
whence a projection of the score function in $x$ via conditional expectation.  Its squared approximation error, relative to the conditional expectation $\widehat{X}_\text{MMSE}:=\operatorname{\mathbb{E}}(X\,\vert\,Y;s)$, then satisfies
\begin{equation*}
(\widehat{X}_\text{MMSE} - \widehat{X})^2 \leq \operatorname{\mathbb{E}}\!\left(X^2 \,\vert\,Y;s\right)
\, \textstyle \operatorname{\mathbb{E}}\!\left(\left[\frac{\partial}{\partial s} \ln p(Y\,\vert\,X;s)\right]^2 \big\vert\,Y;s\right) \! \text{.}
\end{equation*}
\end{thm}
\begin{proof}
Bayes' rule applied to the differential equation of~\eqref{eqn:SkellamPDE} yields the necessary conditional expectations, after which Cauchy-Schwarz serves to bound its latter term. % the term in $\partial / \partial s$.
\end{proof}

While we cannot control the second moment of $X$ conditioned on $Y$ in the bound above, its latter term admits by Property~\ref{prop:SkellamPartials} the equivalence
\begin{equation*}
\textstyle \operatorname{\mathbb{E}}\!\left(\left[\frac{\partial}{\partial s} \ln p(Y\,\vert\,X;s)\right]^2 \big\vert\,Y;s\right) = \operatorname{\mathbb{E}}\!\left(\left[ \frac{(\delta_2^Y \! p)(Y\,\vert\,X;s)}{2p(Y\,\vert\,X;s)} \right]^2 \Big\vert\,Y;s \right)
\text{,}
\end{equation*}
where $\delta_2^Y(\cdot)$ denotes the normalized second central difference in $Y$, analogous to a second derivative.  This term therefore goes as the square of the normalized local curvature in the likelihood at $Y=y$, averaged over the posterior distribution of $X$; it will be small on portions of the domain over which the likelihood remains approximately linear for sets of $X$ having high posterior probability.

Theorem~\ref{thm:SkellamPost} thus provides a means of obtaining Bayesian shrinkage rules under different choices of prior distribution $p(X;s)$, via evaluation of the expectation of~\eqref{eqn:SkellamPost} as
\begin{equation}
\label{eqn:deriv}
\textstyle
p(x\,\vert\,y,s)\big|_{-s}^s -  \operatorname{\mathbb{E}}\left(\frac{\partial}{\partial X} \ln p(X;s) \,\big\vert\,Y;s\right)
\text{.}
\end{equation}
While the above formulation is amenable to further approximation via Taylor expansion (akin to Laplace approximation), we focus here on a direct evaluation of $\operatorname{\mathbb{E}}\left(\frac{\partial}{\partial X} \ln p(X;s) \,\big\vert\,Y;s\right)$.

Discounting the former term of~\eqref{eqn:deriv}, which simply measures the difference in posterior tail decay at $x = \pm s$ and goes to zero with increasing $s$, the derivative on $[-s,s]$ is easily computed for the so-called generalized Gaussian distribution for $p>0$, with location parameter $\mu$ and scale parameter $\sigma_x$:
\begin{align*}
\textstyle
p(x\,;\mu,\sigma_x^2) & \!=\! \frac{1}{2\sigma_x \zeta(p)^{1/p} \, \Gamma(1\!+\!1/p)}
\exp \! \left[- { \textstyle \frac{1}{\zeta(p)} }
\left(\frac{|x-\mu|}{\sigma_x}\right)^{\!p} \, \right] \!
\text{,}
\end{align*}
with $\Gamma(\cdot)$ the Gamma function and $\zeta(p) \!=\! [\Gamma(1/p) / \Gamma(3/p) ]^{p/2}$.

This distribution being unimodal and symmetric about its mean, we obtain for $\mu=0$ the expression
$$
{\textstyle
\operatorname{\mathbb{E}}\left(\frac{\partial}{\partial X} \ln p(X;s) \,\big\vert\,Y;s\right) } =
- \frac{p\,\sigma_x^{-p}}{\zeta(p)}
{\textstyle
\operatorname{\mathbb{E}}\left(\operatorname{sgn}(X) |X|^{p-1} \,\big\vert\,Y;s\right) }
\text{,}
$$
from which the Gaussian ($p=2$) and Laplacian ($p=1$) cases admit straightforward evaluation.

\begin{prop}[Truncated Normal and Laplace Priors]\label{prop:approx_bound}
Let $g(x)$ denote a generalized Gaussian distribution with exponent $p>0$ having mean zero and variance $\sigma_x^2$, %such that $\operatorname{\mathbb{E}}X = 0$ and $\operatorname{Var}X = \sigma_x^2$,
and set  $p(x\,;s) = g(x)\mathbb{I}_{[-s,s]}(x)$.  For $Y \sim \mathcal{S}(X,s)$ we then have:

If $p=2$ so that $p(x\,;s) \propto e^{-x^2/(2\sigma_x^2)}\mathbb{I}_{[-s,s]}(x)$, then
\begin{equation}\label{eqn:GaussImplicitVar}
\operatorname{Var}X = 2\sigma_x^2 \frac{\gamma(3/2, s^2/2\sigma_x^2)}{\gamma(1/2, s^2/2\sigma_x^2)}
= \sigma_x^2 \!-\! \frac{\sqrt{2}\,s\sigma_x e^{-s^2/2\sigma_x^2}} {\sqrt{\pi} \operatorname{erf} ( s/ \sqrt{2}\sigma_x)}
\text{,}
\end{equation}
with $\gamma(\cdot,\cdot)$ the lower incomplete Gamma function, and
\begin{equation*}
\textstyle
\operatorname{\mathbb{E}}\left(\frac{\partial}{\partial X} \ln p(X;s) \,\big\vert\,Y;s\right) =
- \sigma_x^{-2} \operatorname{\mathbb{E}}(X\,\vert\,Y;s)
\text{.}
\end{equation*}

If $p=1$ so that $p(x\,;s) \propto e^{-|x|/(\sigma_x/\sqrt{2})}\mathbb{I}_{[-s,s]}(x)$, then
\begin{equation}\label{eqn:LaplaceImplicitVar}
\operatorname{Var}X = \frac{\sigma_x^2}{2} \frac{\gamma(3, \sqrt{2}s/\sigma_x)}{\gamma(1, \sqrt{2}s/\sigma_x)}
=  \sigma_x^2 - \frac{s(s+\sqrt{2}\sigma_x) e^{-\sqrt{2}s/\sigma_x}}{1 - e^{-\sqrt{2}s/\sigma_x}}
\text{;}
\end{equation}
\begin{equation*}
\textstyle
\operatorname{\mathbb{E}}\left(\frac{\partial}{\partial X} \ln p(X;s) \,\big\vert\,Y;s\right) =
- \sqrt{2}\sigma_x^{-1} \operatorname{\mathbb{E}}(\operatorname{sgn}(X)\,\big\vert\,Y;s)
\text{.}
\end{equation*}
\end{prop}

Combining Proposition~\ref{prop:approx_bound} with Theorem~\ref{thm:SkellamPost} yields approximate posterior mean estimators under truncated Gaussian and Laplacian priors.  The Gaussian case recovers the shrinkage rule
\begin{equation}
\label{eqn:estimator-Gaussian}
\widehat{X} = \frac{\sigma_x^2}{s+\sigma_x^2} \, Y \text{,}
\end{equation}
the optimal linear estimator under a second-moment Normal approximation to the Skellam likelihood,  with mean zero and variance $s$.  The heavier-tailed Laplacian case yields an implicit shrinkage rule illustrated in Fig.~\ref{fig:posterior},
\begin{figure}[!t]
\begin{center}
\includegraphics[width=0.525\textwidth]{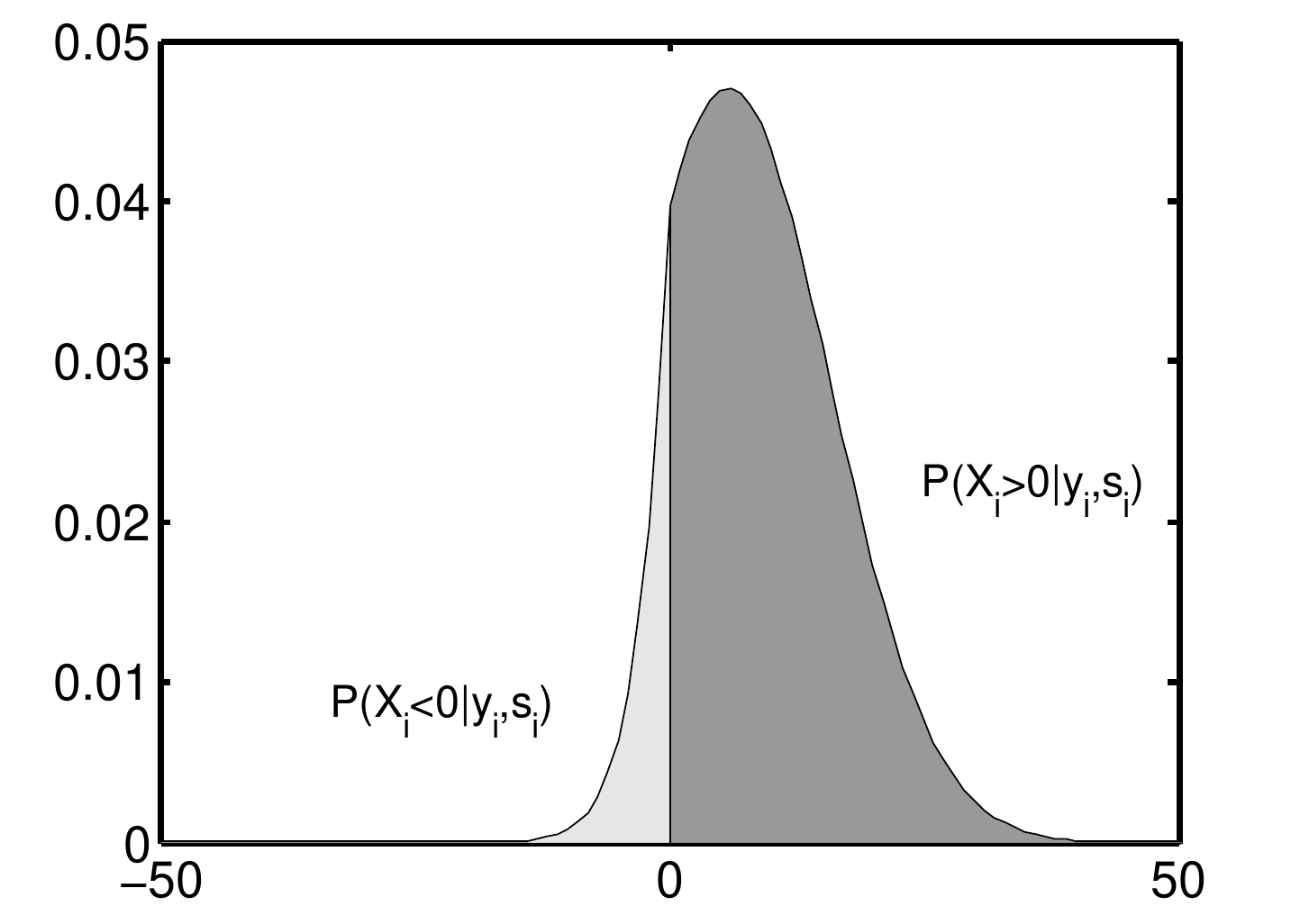}
\caption{\label{fig:posterior} Illustration of the shrinkage implicit in~\eqref{eqn:estimator-Laplace} as a function of the posterior distribution $p(y\,\vert\,x;s)$ in the case of a Skellam likelihood, showing the contribution of posterior mass to shrinkage toward and away from zero}
\end{center}
\end{figure}
whose asymptotic behavior in turn enables a simple soft-thresholding rule to be fitted:
\begin{align}
\label{eqn:estimator-Laplace} %SB
\widehat{X} & = Y - \frac{\sqrt{2}s}{\sigma_x} \left[ \operatorname{Pr}(X > 0\,\vert\,Y;s) - \operatorname{Pr}(X < 0\,\vert\,Y;s) \right]
\\ & \approxeq \operatorname{sgn}(Y) \, \max \! \left( |Y| - \sqrt{2}s / \sigma_x, 0 \right)
\label{eqn:SBT}
\text{,}
\end{align}
The soft-thresholding estimator of~\eqref{eqn:SBT} can in turn be adapted to yield a piecewise-linear estimator whose slope matches that of~\eqref{eqn:estimator-Laplace} at the origin.  To accomplish this, note that for any prior distribution with even symmetry,~\eqref{eq:skellamPMF} of Definition~\ref{def:skellam} implies odd symmetry of the posterior expectation functional; i.e., $\operatorname{\mathbb{E}}(X\,\vert\,Y=y\,;s)  = - \operatorname{\mathbb{E}}(X\,\vert\,Y=-y\,;s) \text{.}$  Therefore the slope of any shrinkage estimator at the origin may be computed as
\begin{equation*}
\frac{1}{2}\left[ \operatorname{\mathbb{E}}(X\,\vert\,Y\!=\!1\,;s) - \operatorname{\mathbb{E}}(X\,\vert\,Y\!=\!-1\,;s) \right]  = \operatorname{\mathbb{E}}(X\,\vert\,Y\!=\!1\,;s) \text{.}
\end{equation*}
The slope term $\operatorname{\mathbb{E}}(X\,\vert\,Y\!=\!1\,;s)$ may in turn be pre-computed to arbitrary accuracy using numerical methods, and indexed as a function of $s$ and prior variance $\sigma_x^2$, yielding the following piecewise-linear shrinkage estimator:
\begin{equation}
\label{eqn:SBL}
\widehat{X} \!=\! \operatorname{sgn}(Y) \max \! \left( |Y| \!-\! \sqrt{2}s / \sigma_x, \, \operatorname{\mathbb{E}}(X\,\vert\,Y\!=\!1\,;s)\, |Y| \right) \!
\text{.}
\end{equation}
Figure~\ref{fig:bayes}
\begin{figure}
\centering
%\subfigure[Bayesian estimators under a Normal prior]{\includegraphics[width=.4\textwidth]{BayesGaussian}}
%\subfigure[Bayesian estimators under a Laplace prior]{}
\includegraphics[width=.525\textwidth]{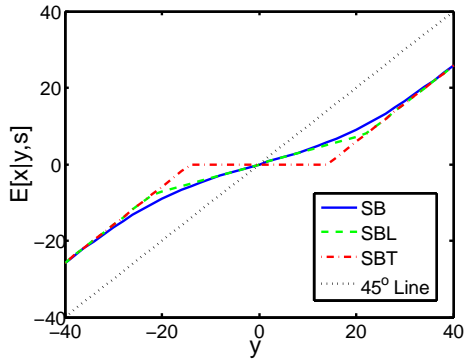}
\caption{\label{fig:bayes}Bayesian shrinkage rules corresponding to a Laplacian prior and Skellam likelihood, with dotted 45$^\circ$ line shown for reference: Skellam Bayes (SB) MMSE shrinkage rule, computed numerically; soft-thresholding (SBT) approximation of~\eqref{eqn:SBT}; and piecewise-linear (SBL) approximation of~\eqref{eqn:SBL}}
\end{figure}
in turn compares the exact posterior mean shrinkage rule, corresponding to $\operatorname{\mathbb{E}}(X\,\vert\,Y;s)$ and computed numerically, with the soft-thresholding estimator of~\eqref{eqn:SBT} and the piecewise-linear estimator of~\eqref{eqn:SBL}.  The ideas above can be straightforwardly extended to the multivariate case~\cite{ref:HirakawaWolfe09b}, owing to conditional independence properties of the Skellam likelihood; derivatives may also be computed for the case of mixture priors, though no efficient solution is yet known to compute the mixture weights.

\subsection{Parameter and Risk Estimation for Skellam Shrinkage}
\label{sec:modified_SURE}

Having derived Bayes estimators for the class of unimodal, zero-mean, symmetric priors considered above, we now turn to parameter and risk estimation for Skellam shrinkage.  With only a single observation of each Haar coefficient in this heteroscedastic setting, maximum-likelihood methods will simply return the identity as a shrinkage rule.  However, by borrowing strength across multiple coefficient observations we may improve upon the risk properties of this approach; as we now detail, this is equally attainable in a  frequentist or Bayes setting.  Here we consider coefficient aggregation within a given scale, with notation $\sum_i (\cdot)_i$ below indicating summation over location parameter $i$ within a single Haar subband.

The main result of this section is the following theorem, which yields a procedure for unbiased $\ell^2$ risk estimation in the context of soft thresholding and other shrinkage operators.
\begin{thm}[Unbiased Risk Estimation]\label{thm:skellam}
Let $y_i \sim \mathcal{S}(x_i,s_i)$ %for $i \in \{1,2\ldots,N\}$,
and $t_i=y_i\p + y_i\m$ according to~\eqref{eqn:ObsWav}, with $x_i,s_i$ unknown.  Fix a vector-valued estimator $\widehat{\bm{X}}(\bm{Y}\!,\bm{T})=\bm{Y}+\bm{\theta}(\bm{Y}\!,\bm{T})$, where $\bm{\theta}:\mathbb{Z}^N\times\mathbb{Z}_+^N\to\mathbb{R}^N$, and let $\bm{1}$ denote the vector of all ones.   Then the $\ell^2$ risk of $\widehat{\bm{X}}(\bm{Y}\!,\bm{T})$ may be expressed
%in terms of $\bm{Y}$ and $\bm{T}$
as follows:
\begin{multline}
\label{eqn:modified_SURE1}
\operatorname{\mathbb{E}} \| \widehat{\bm{X}}(\bm{Y}\!,\bm{T}) - \bm{x} \|_2^2
= \operatorname{\mathbb{E}}\!\Big[ \| \bm{\theta}(\bm{Y}\!,\bm{T}) \|_2^2 + \| \bm{T} \|_1 +
2 \langle \bm{Y} \!, \bm{\theta}(\bm{Y}\!,\bm{T}) \rangle \\
- \! \langle \bm{T}\!\!+\!\bm{Y}\!,\bm{\theta}(\bm{Y}\!\!-\!\bm{1},\bm{T}\!-\!\bm{1}) \rangle \!+\! \langle \bm{T}\!-\!\bm{Y} \!, \bm{\theta}(\bm{Y}\!\!+\!\bm{1},\bm{T}\!-\!\bm{1}) \rangle \Big] \text{,}
\end{multline}
with
$\| \bm{\theta}(\bm{y},\bm{t}) \|_2^2 \!+\! \| \bm{t} \|_1 \!+\!
2 \langle \bm{y} \!, \bm{\theta}(\bm{y},\bm{t}) \rangle - \!\langle \bm{t}\!\!+\!\bm{y}\!,\bm{\theta}(\bm{y}\!\!-\!\bm{1},\bm{t}\!-\!\bm{1}) \rangle \!+\! \langle \bm{t}\!-\!\bm{y} \!, \bm{\theta}(\bm{y}\!\!+\!\bm{1},\bm{t}\!-\!\bm{1}) \rangle$
an unbiased estimate thereof.
\end{thm}
\begin{proof}
The risk $\operatorname{\mathbb{E}} \| \widehat{\bm{X}}(\bm{Y}\!,\bm{T}) - \bm{x} \|_2^2$ may be expanded as
\begin{equation}
\label{eq:SkelRiskExpansion}
\operatorname{\mathbb{E}} \| \bm{\theta}(\bm{Y}\!,\bm{T}) \|_2^2 + \operatorname{\mathbb{E}} \|\bm{Y}-\bm{x}\|_2^2 + 2 \operatorname{\mathbb{E}} \langle \bm{Y}\!\!-\!\bm{x},\bm{\theta}(\bm{Y}\!,\bm{T}) \rangle \text{,}
\end{equation}
with $\operatorname{\mathbb{E}} \|\bm{Y}-\bm{x}\|_2^2 = \sum_i \operatorname{Var}Y_i = \sum_i s_i = \operatorname{\mathbb{E}} \| \bm{T} \|_1$.  To evaluate the final term in~\eqref{eq:SkelRiskExpansion} above, note first that
$ \bm{Y} - \bm{x} = (\bm{Y\p} \!-\! \bm{Y\m}) - (\bm{x\p} \!-\! \bm{x\m})$ according to~\eqref{eqn:wavsca} and~\eqref{eqn:ObsWav}, and furthermore that $\bm{Y^{\pm}} = (\bm{T} \pm \bm{Y})/2 $.  By conditioning on $\bm{Y\p}$ or $\bm{Y\m}$ we in turn obtain Poisson variates, and thus general results for discrete exponential families~\cite{ref:Hudson_1978, ref:Ghosh_1983} apply, yielding the final relations needed to complete the proof of Theorem~\ref{thm:skellam}:
\begin{equation*}
\operatorname{\mathbb{E}} \langle \bm{Y^{\pm}}-\bm{x^{\pm}} , \bm{\theta}(\bm{Y}\!,\bm{T}) \rangle =
\operatorname{\mathbb{E}} \langle \bm{Y^{\pm}} , \bm{\theta}(\bm{Y}\!,\bm{T}) -
\bm{\theta}(\bm{Y} \mp \bm{1},\bm{T}-\bm{1}) \rangle \text{.}
\end{equation*}
\end{proof}

Parameters of any chosen estimator form $\widehat{\bm{X}}(\bm{Y}\!,\bm{T})$ may thus be optimized by minimizing the unbiased risk estimate of Theorem~\ref{thm:skellam} with respect to observed data vectors $\bm{y}$ and $\bm{t}$. As an important special case, we obtain the following corollary.
\begin{cor}[SkellamShrink]\label{cor:skellamshrink}
The optimal threshold $\tau$ for soft thresholding as $\hat{X}_i({Y}_i; \tau) := \operatorname{sgn}({Y}_i)\max(|{Y}_i|-\tau,0)$
is obtained by minimizing
\begin{equation}\label{eqn:skellamshrink}
\sum_i \operatorname{sgn}(|y_i|-\tau) t_i
\,+ \sum_i \min(y_i^2,\tau^2)-\tau \#\{i:|y_i|=\tau\}
\text{.}
\end{equation}
\end{cor}

\begin{rem}
\end{rem}\vspace{-0.75\baselineskip}
Recall the Stein's unbiased risk estimate SUREshrink result~\cite{ref:Donoho_1995} for soft thresholding in the case of additive white Gaussian noise of variance $\sigma^2$, as described in~\eqref{eqn:softThresh}--\eqref{eqn:SURE} of Section~\ref{sec:review_denoising}.  Recasting the objective function of~\eqref{eqn:SURE} for SUREshrink threshold optimization as
\begin{equation}\label{eqn:SURE2}
\hspace{-8.75em}\sum_i \operatorname{sgn}(|y_i|\!-\!\tau)\sigma^2 \! + \! \sum_i\min(y_i^2,\tau^2)\text{,}
\end{equation}
we see that $t_i$ in~\eqref{eqn:skellamshrink} plays a role analogous to $\sigma^2$ in the homoscedastic SUREShrink setting represented by~\eqref{eqn:SURE2}, with the dependence on coefficient index $i$ reflecting the heteroscedasticity present in the Skellam likelihood case. %, additive white Gaussian noise setting of SUREShrink~\cite{ref:Donoho_1995}, reinforcing the importance of heteroscedasticity in the Poisson case.

\subsubsection{SkellamShrink with Adjusted Thresholds}

We may also consider a generalization of the SkellamShrink soft thresholding estimator of Corollary~\eqref{cor:skellamshrink}, inspired by the Bayes point estimator $\operatorname{sgn}({Y}_i)\max(|{Y}_i|- \tau(s_i),0)$ of Theorem~\ref{thm:SkellamPost}, in which individual coefficient thresholds depend in general on the corresponding scaling coefficient.  By treating the quantity $\sigma_x$ appearing in the Bayesian estimators of Section~\ref{sec:BayesEst} not as a prior variance parameter, but simply as part of a parametric risk form to be optimized, we may appeal directly to the unbiased risk estimation formulation of Theorem~\ref{thm:skellam}. Since a priori knowledge limitations may well preclude exact prior elicitation in practice, this flexible approach provides a degree of robustness to prior model mismatch, as borne out by our simulation studies below.

As an example, consider a shrinkage estimator $\hat{X}_{i}(Y_i,T_i)=y_i+\theta(Y_i,T_i; \sigma_x)$ that depends on $T_i$ and unknown parameter $\sigma_x$ % (or precision $\lambda = \sqrt{2}/\sigma_x$ ***)
as per the soft thresholding formulation of~\eqref{eqn:SBT}:
\begin{equation}
\label{eqn:SkellamAdjusted}
\theta(Y_i,T_i; \sigma_x) = \begin{cases}
-\operatorname{sgn}(Y_i)\frac{\sqrt{2}}{\sigma_x}T_i& \text{if $|Y_i|\geq \frac{\sqrt{2}}{\sigma_x}T_i$}\\
-Y_i & \text{if $|Y_i| < \frac{\sqrt{2}}{\sigma_x}T_i$;}
\end{cases}
\end{equation}
%Note that we may assume $\sigma_x> \sqrt{2}$, since $|Y_i| < T_i$ by their respective definitions.
Defining $\tilde{\sigma}_x = \sigma_x / \sqrt{2}$ and $\tilde{t}_i=(t_i-1)/\tilde{\sigma}_x$ for notational convenience, we have the risk estimate
\begin{align*}
%& R_{\text{SH}}(y,t;\sigma_x)
& \,\, \sum_i \operatorname{sgn}\left(|y_i|-\tilde{t}_i\right) t_i
+\sum_i \min\left(y_i^2,{t_i^2}/{\tilde{\sigma}_x^2}\right)
\\
& \qquad \qquad \qquad -2\!\!\!\!\sum_{\!\!\!\!\!i:|y_i|>{t_i}/{\tilde{\sigma}_x}\!\!\!\!} \sqrt{{t_i^2}/{\tilde{\sigma}_x^2}} %\frac{|y_i|}/{\tilde{\sigma}_x}
+\sum_{\!\!\!\!\!i:|y_i|=\lceil\tilde{t}_i-1\rceil\!\!\!\!\!} c(t_i)
+\sum_{\!\!\!\!\!i:|y_i|=\lfloor\tilde{t}_i+1\rfloor\!\!\!\!\!} d(t_i)
\text{;}
\\
c(t_i) & = t_i\left(\lceil \tilde{t}_i\rceil
-\tilde{t}_i\right)-\lceil (\tilde{t}_i-1)\rceil^2+\left(\tilde{t}_i-1\right)\lceil
(\tilde{t}_i-1)\rceil \text{,}
\\
d(t_i) & = \begin{cases}
t_i(\lfloor\tilde{t}_i\rfloor-\tilde{t}_i)-\lfloor
(\tilde{t}_i+1)\rfloor^2+\left(\tilde{t}_i-1\right)\lfloor(\tilde{t}_i+1)\rfloor
& \text{if ${t_i}/{\tilde{\sigma}_x} \geq \lfloor\tilde{t}_i+1\rfloor$} \\
t_i(\lfloor\tilde{t}_i\rfloor-\tilde{t}_i)+\lfloor(\tilde{t}_i+1)\rfloor^2-\left(\tilde{t}_i-1\right)\lfloor(\tilde{t}_i+1)\rfloor
& \text{if ${t_i}/{\tilde{\sigma}_x}<\lfloor\tilde{t}_i+1\rfloor$.}
\end{cases}
\end{align*}
with $\lfloor \cdot \rfloor$ and $\lceil \cdot \rceil$ denoting the floor and ceiling operators, respectively, and $c(t_i)$ and $d(t_i)$ adjusting for the singularity at $|y_i|=\lceil\tau_i\rceil\pm 1$.

\subsubsection{Unbiased Risk Estimates for Variance-Stabilized Shrinkage}

The strategy outlined above naturally generalizes to any form of parametric estimator via the unbiased risk estimation formulation of Theorem~\ref{thm:skellam}, enabling an improvement over the variance-stabilization strategies of Section~\ref{sec:intro} by direct minimization of empirical risk.  As a specific example, consider the Haar-Fisz estimator of~\cite{ref:Fryzlewicz_2004}, in which each empirical Haar wavelet coefficient $y_i$ is scaled by the root of its corresponding empirical scaling coefficient as  $\tilde{y}_i := y_i/\sqrt{t_i}$ in order to achieve variance stabilization, after which standard Gaussian shrinkage methods such as SUREShrink are applied and the variance stabilization step inverted.

For the case of nonlinear shrinkage operators, of course, neither the resultant estimators nor the risk estimates themselves will in general commute with this Haar-Fisz strategy, leading to a loss of the unbiasedness property of risk minimization---in contrast to the direct application of Theorem~\ref{thm:skellam}.  Taking Haar-Fisz soft thresholding with some fixed threshold $\tau$ as an example, the equivalent Skellam shrinkage rule is seen to be $\hat{X}_i(Y_i, T_i; \tau) = \operatorname{sgn}({Y}_i)\max(|{Y}_i|-\sqrt{T_i}\,\tau,0)$---with $\sqrt{T_i}$ in contrast to the scaling of $T_i$ implied by Theorem~\ref{thm:SkellamPost}, as in the adjusted-threshold approach of~\eqref{eqn:SkellamAdjusted} above.  In an analogous manner, the corresponding \emph{exact} unbiased risk estimate for this shrinkage rule can in turn be derived directly by appeal to Theorem~\ref{thm:skellam}, rather than relying on the heretofore standard Haar-Fisz approach of SUREShrink empirical risk minimization via~\eqref{eqn:SURE2}, applied to the variance-stabilized coefficients $\tilde{y}_i$.

\subsubsection{Empirical Bayes via Method of Moments}
\label{sec:EBayes}

We conclude this section with a simple and effective empirical Bayes strategy for estimating scaling coefficients $\{s_i\}$ and prior parameter $\sigma_x$ for the Bayesian shrinkage rules derived in Section~\ref{sec:BayesEst} above.  Recall from~\eqref{eqn:ypm} that $s_i = \operatorname{\mathbb{E}}T_i$, implying the use of the empirical scaling coefficient $t_i$ as a direct substitute for $s_i$ in the Bayesian setting.  Note that $s_i = \sum_{j:\,|W_{ij}|=1} f_j$ for Haar transform matrix $\bm{W}$, with $T_i$ a corresponding sum of Poisson variates with means $f_j$ representing the underlying intensities of interest to be estimated.  In turn, as the sum $s_i$ increases, the relative risk $\operatorname{\mathbb{E}} | T_i - s_i |^2 / \,s_i^2$ of the plug-in estimator $\hat{s}_i = T_i$ will rapidly go to zero precisely at rate $1/s_i$.

Next note that under the assumption of a unimodal, zero-mean, and symmetric prior distribution $p(X;s)$, only $\operatorname{Var}X$ remains to be estimated.  A convenient moment estimator is available, since
$T_i\sim \mathcal{P}(s_i)$ and $Y_i\sim \mathcal{S}(x_i,s_i)$ together imply that $\operatorname{Var}T_i = \operatorname{Var}Y_i = \operatorname{\mathbb{E}}Y_i^2 - X_i^2$, and hence we obtain $\widehat{\operatorname{Var}}X = (1/N)\sum_i y_i^2-t_i$.  Once estimates $\widehat{\operatorname{Var}}X$ and $\{s_i\}$ are obtained for the coefficient population of interest, the implicit variance equations of~\eqref{eqn:GaussImplicitVar} and~\eqref{eqn:LaplaceImplicitVar} may be solved numerically to yield scale parameter $\sigma_x$ of the truncated generalized Gaussian distribution considered earlier, with
%$\sigma_x^2 = \lim_{s\to\infty} \operatorname{Var}X$.
%$\sigma_x^2 \rightarrow \operatorname{Var}X$ as $s \rightarrow \infty$.
$\sigma_x^2 = \operatorname{Var}X$ in the limit as $s$ grows large.  In our simulation regimes, we observed no discernable difference in overall wavelet-based estimation performance by setting $\sigma_x^2 = \widehat{\operatorname{Var}}X$ directly. %, again reflecting the fact that $s_i$ is typically ``large enough.''

\section{Simulation Studies}
\label{sec:results}

We now describe a series of simulation studies undertaken to evaluate the efficacy of the wavelet-based shrinkage estimators derived above.  We considered exact Skellam Bayes (SB) posterior mean estimators, computed numerically with respect to a given prior; the Skellam Bayes Gaussian approximation (SBG) linear shrinkage of~\eqref{eqn:estimator-Gaussian}; the Skellam Bayes Laplacian soft-thresholding (SBT) approximation of~\eqref{eqn:SBT}; the Skellam Bayes Laplacian piecewise-linear (SBL) approximation of~\eqref{eqn:SBL}; the SkellamShrink (SS) soft-thresholding estimator with empirical risk minimization of Corollary~\ref{cor:skellamshrink}; and the SkellamShrink hybrid (SH) adjusted-threshold shrinkage of~\eqref{eqn:SkellamAdjusted}.

Estimators were implemented using a 3-level undecimated Haar wavelet decomposition, with empirical risk minimization or the moment methods of Section~\ref{sec:EBayes} above used to estimate parameters for the corresponding shrinkage rules.  As first comparison of relative performance, Figs.~\ref{fig:boxplot1} and~\ref{fig:boxplot2} tabulate results in mean-squared error (MSE) for Skellam likelihood inference in cases when the latent variables of interest are drawn from Normal and Laplacian distributions with known parameter $\sigma_x^2 \in \{32, 64, 128\}$. The accompanying box plots are shown on a log-MSE scale for visualization purposes, in order to better reveal differences between estimator performance.  These figures confirm that exact Bayesian estimators (SB) outperform all others, but indicate that prior-specific Skellam Bayes approximations SBG and SBL are comparable, respectively, for the Gaussian and Laplacian cases over the range of prior parameters shown here.  Among soft-thresholding approaches, the frequentist SkellamShrink methods SS and SH in turn offer improvements over the Bayesian soft-thresholding estimator SBT.
\begin{figure*}
\begin{center}
{\footnotesize
\begin{tabular}{c|rrrrr|rrrrr}
Gaussian & \multicolumn{5}{c|}{(a)~$\sigma_x^2=32$} & \multicolumn{5}{c}{(b)~$\sigma_x^2=64$}\\
MSE &\multicolumn{1}{c}{SS}
 &\multicolumn{1}{c}{SB}
 &\multicolumn{1}{c}{SBG}
 &\multicolumn{1}{c}{SBT}
 &\multicolumn{1}{c|}{SH}
 &\multicolumn{1}{c}{SS}
 &\multicolumn{1}{c}{SB}
 &\multicolumn{1}{c}{SBG}
 &\multicolumn{1}{c}{SBT}
 &\multicolumn{1}{c}{SH}
\\
\hline
Mean&30.22 & 24.03 & 24.03 & 31.32 & 29.98 & 47.75 & 39.01 & 39.01 & 53.42 &
48.28\\
Median&13.53 & 11.14 & 11.13 & 14.45 & 13.88 & 21.66 & 18.32 & 18.33 & 24.57 &
21.95\\
Std.~Dev.&44.65 & 33.20 & 33.20 & 43.58 & 42.05 & 68.22 & 54.99 & 54.99 & 73.92 &
67.92\\
\end{tabular}\\\begin{tabular}{c|rrrrr|rrrrr}
 & \multicolumn{5}{c}{(c)~$\sigma_x^2=128$}\\
 &\multicolumn{1}{c}{SS}
 &\multicolumn{1}{c}{SB}
 &\multicolumn{1}{c}{SBG}
 &\multicolumn{1}{c}{SBT}
 &\multicolumn{1}{c|}{SH}
 \\
\hline
Mean& 66.71 & 56.11 & 56.11 & 73.76 & 67.02 \\
Median& 29.88 & 25.21 & 25.20 & 34.35 & 30.38 \\
Std.~Dev.& 97.04 & 80.20 & 80.22 & 103.29 & 96.12 \\
\end{tabular}
}\\\subfigure[$\sigma_x^2=32$]{\includegraphics[width=.32\textwidth]{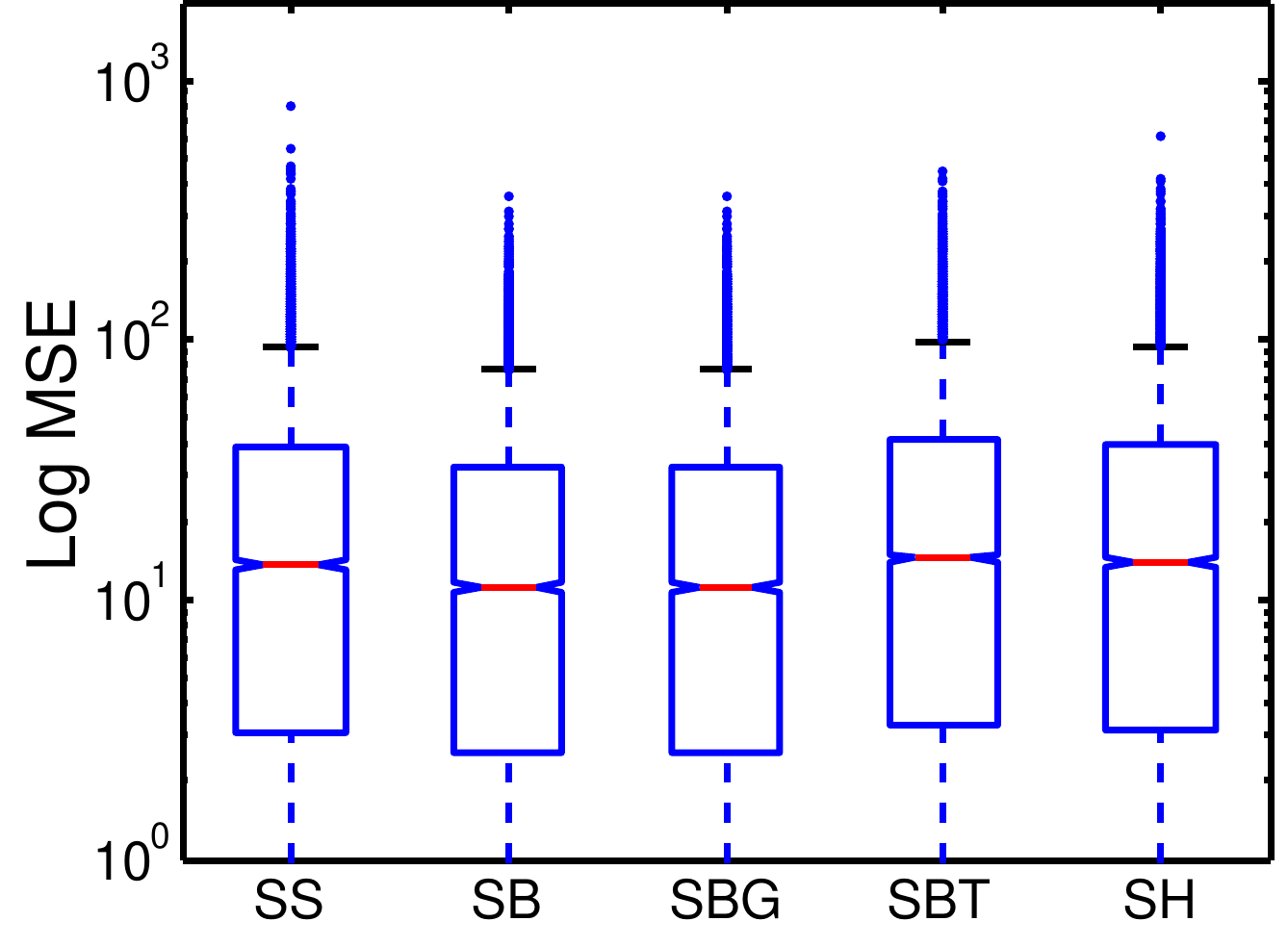}}
\subfigure[$\sigma_x^2=64$]{\includegraphics[width=.32\textwidth]{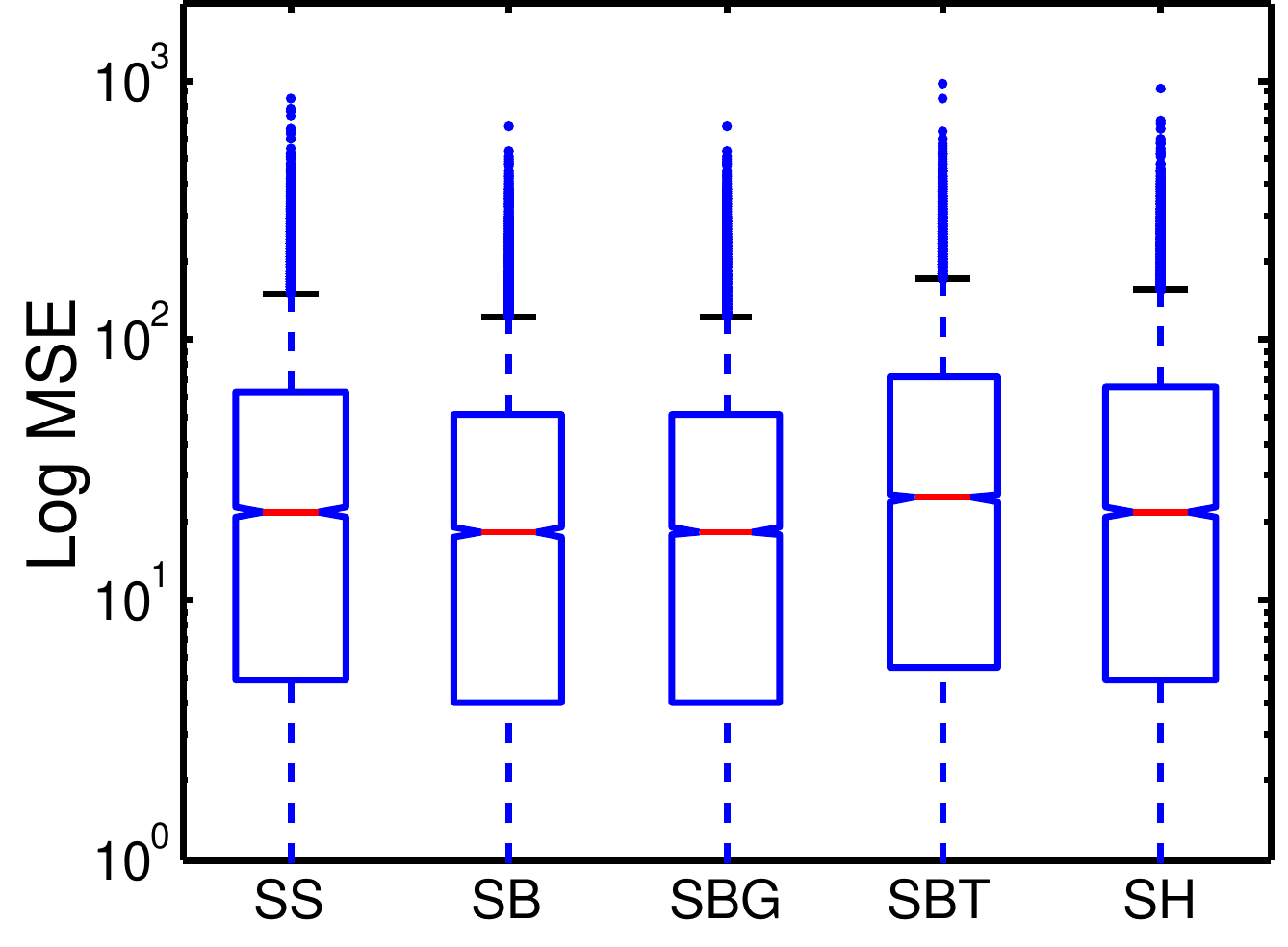}}
\subfigure[$\sigma_x^2=128$]{\includegraphics[width=.32\textwidth]{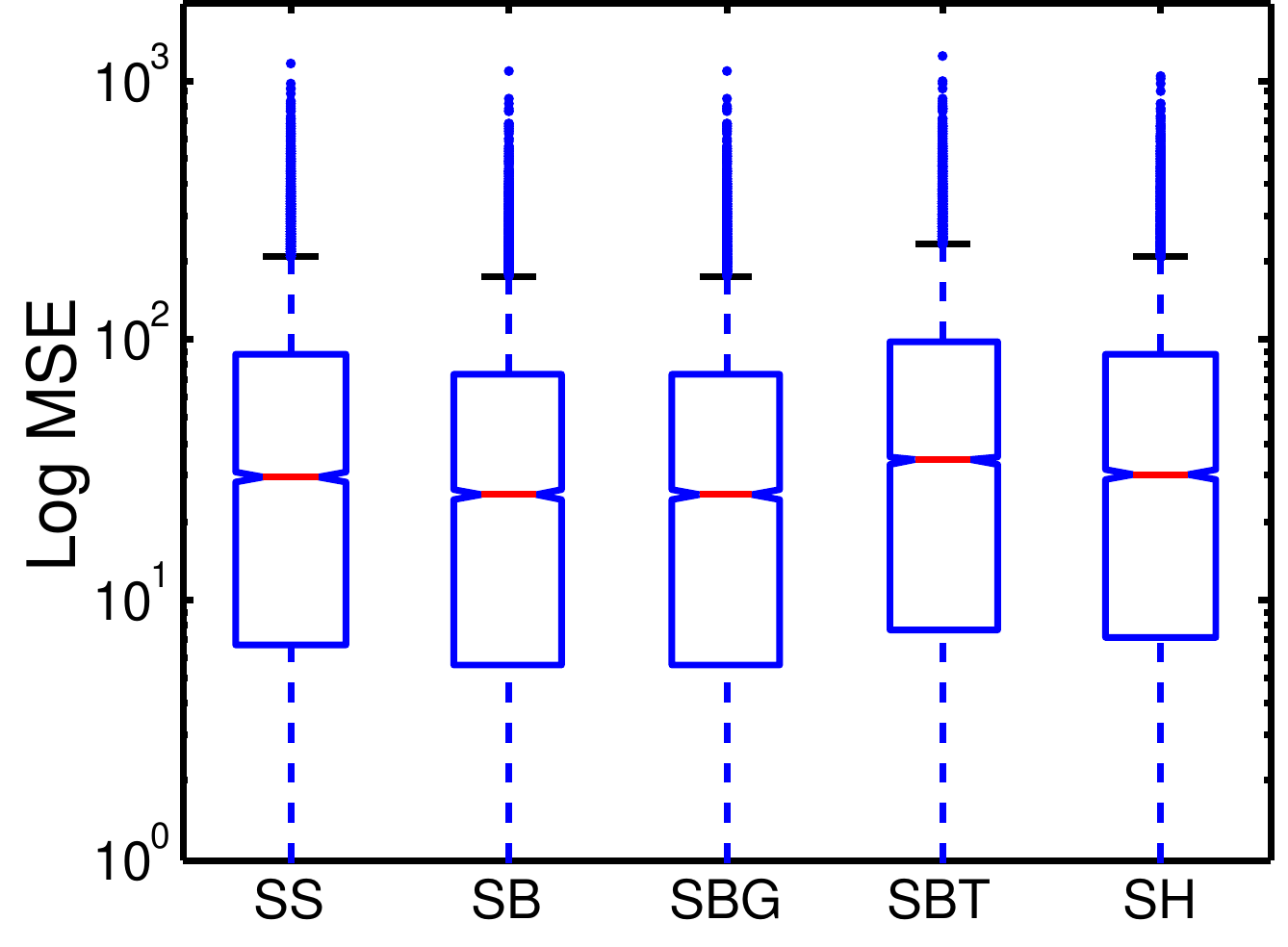}}
\caption{\label{fig:boxplot1}Empirical performance as measured by MSE, with $x$ drawn from a truncated Normal
  distribution and scaling coefficient $s$ fixed to $100$}
\end{center}
\end{figure*}
\begin{figure*}
\begin{center}
{\footnotesize
\begin{tabular}{c|rrrrr|rrrrr}
Laplacian & \multicolumn{5}{c|}{(a)~$\sigma_x^2=32$} & \multicolumn{5}{c}{(b)~$\sigma_x^2=64$} \\
MSE &\multicolumn{1}{c}{SS}
 &\multicolumn{1}{c}{SB}
 &\multicolumn{1}{c}{SBL}
 &\multicolumn{1}{c}{SBT}
 &\multicolumn{1}{c|}{SH}
 &\multicolumn{1}{c}{SS}
 &\multicolumn{1}{c}{SB}
 &\multicolumn{1}{c}{SBL}
 &\multicolumn{1}{c}{SBT}
 &\multicolumn{1}{c}{SH}\\
\hline
Mean&27.64 & 24.05 & 24.37 & 29.97 & 27.77 & 42.90 & 37.98 & 38.44 & 46.78 &
42.75\\
Median& 7.26 & 7.53 & 7.21 & 7.12 & 7.10 & 13.96 & 13.59 & 13.11 & 13.84 &
13.66\\
Std.~Dev.&55.11 & 46.96 & 49.92 & 64.34 & 56.07 & 75.24 & 66.12 & 69.74 & 86.44 &
75.74\\
\end{tabular}
\begin{tabular}{c|rrrrr|}
 & \multicolumn{5}{c}{(c)~$\sigma_x^2=128$}\\
 &\multicolumn{1}{c}{SS}
 &\multicolumn{1}{c}{SB}
 &\multicolumn{1}{c}{SBL}
 &\multicolumn{1}{c}{SBT}
 &\multicolumn{1}{c|}{SH}\\
\hline
Mean& 59.63 & 54.51 & 55.18 & 64.19 & 59.91 \\
Median& 21.65 & 20.94 & 20.47 & 22.13 & 21.63 \\
Std.~Dev.& 96.08 & 86.43 & 89.60 & 105.59 & 96.87 \\
\end{tabular}
}\\\subfigure[$\sigma_x^2=32$]{\includegraphics[width=.32\textwidth]{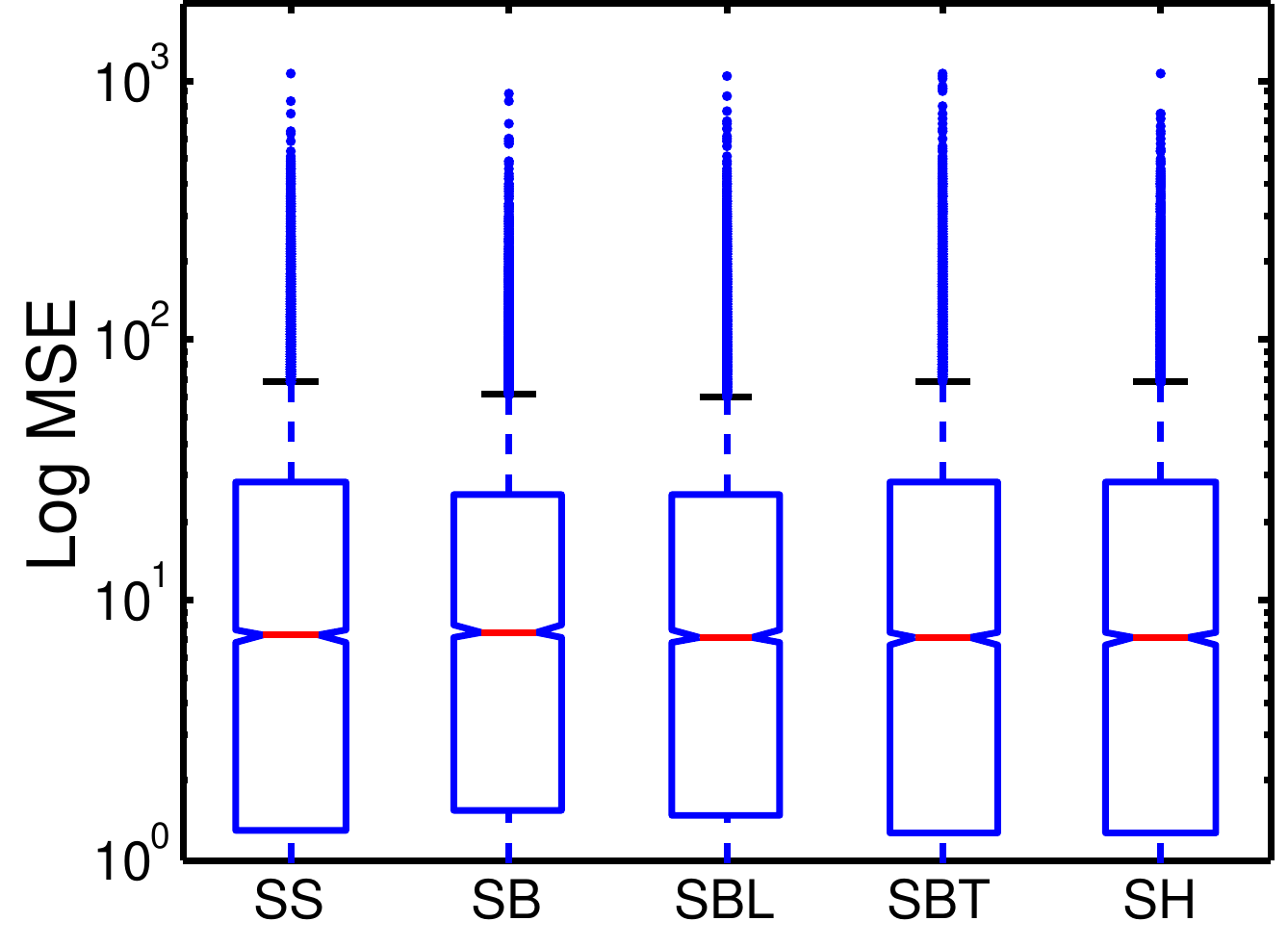}}
\subfigure[$\sigma_x^2=64$]{\includegraphics[width=.32\textwidth]{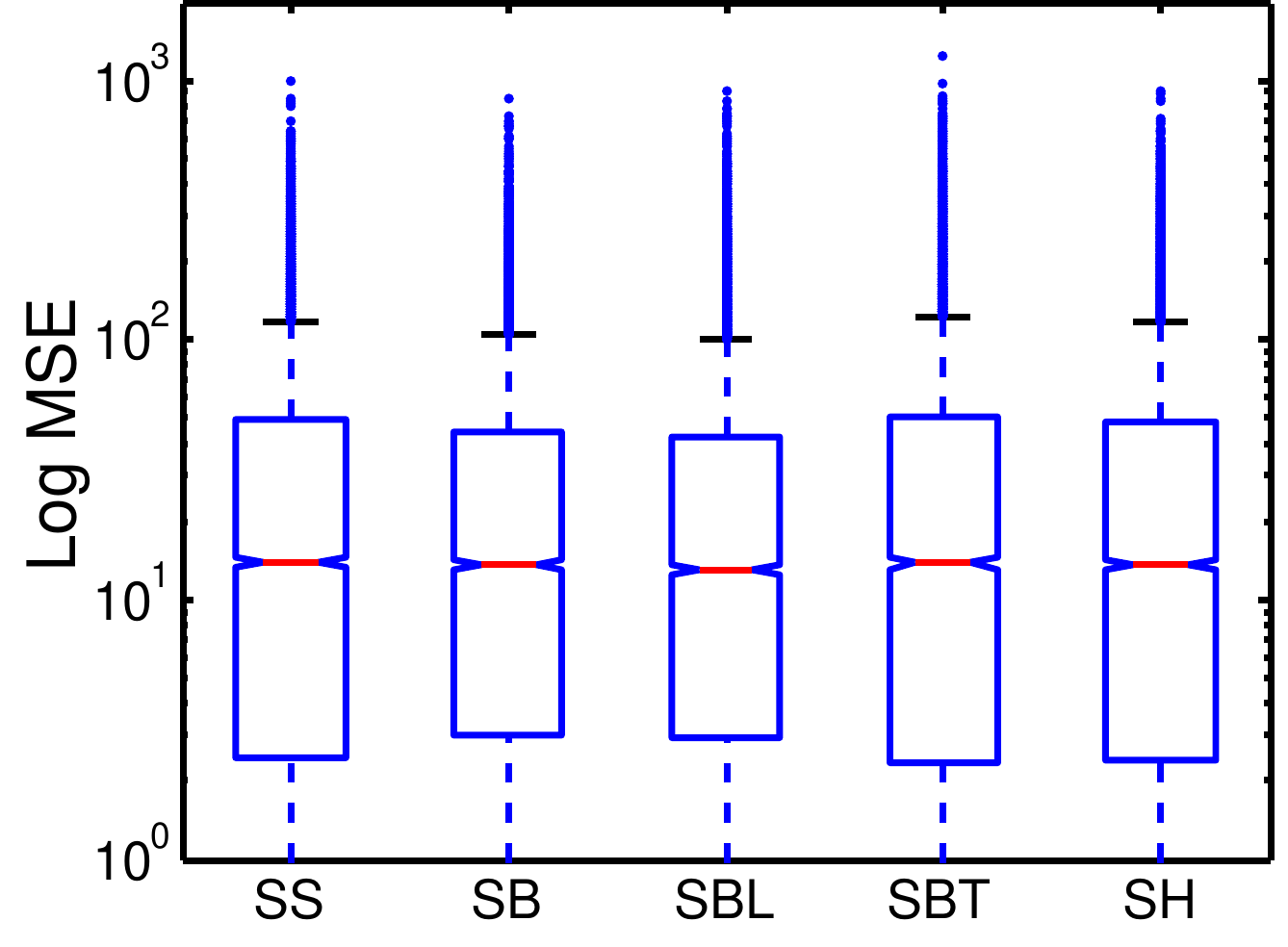}}
\subfigure[$\sigma_x^2=128$]{\includegraphics[width=.32\textwidth]{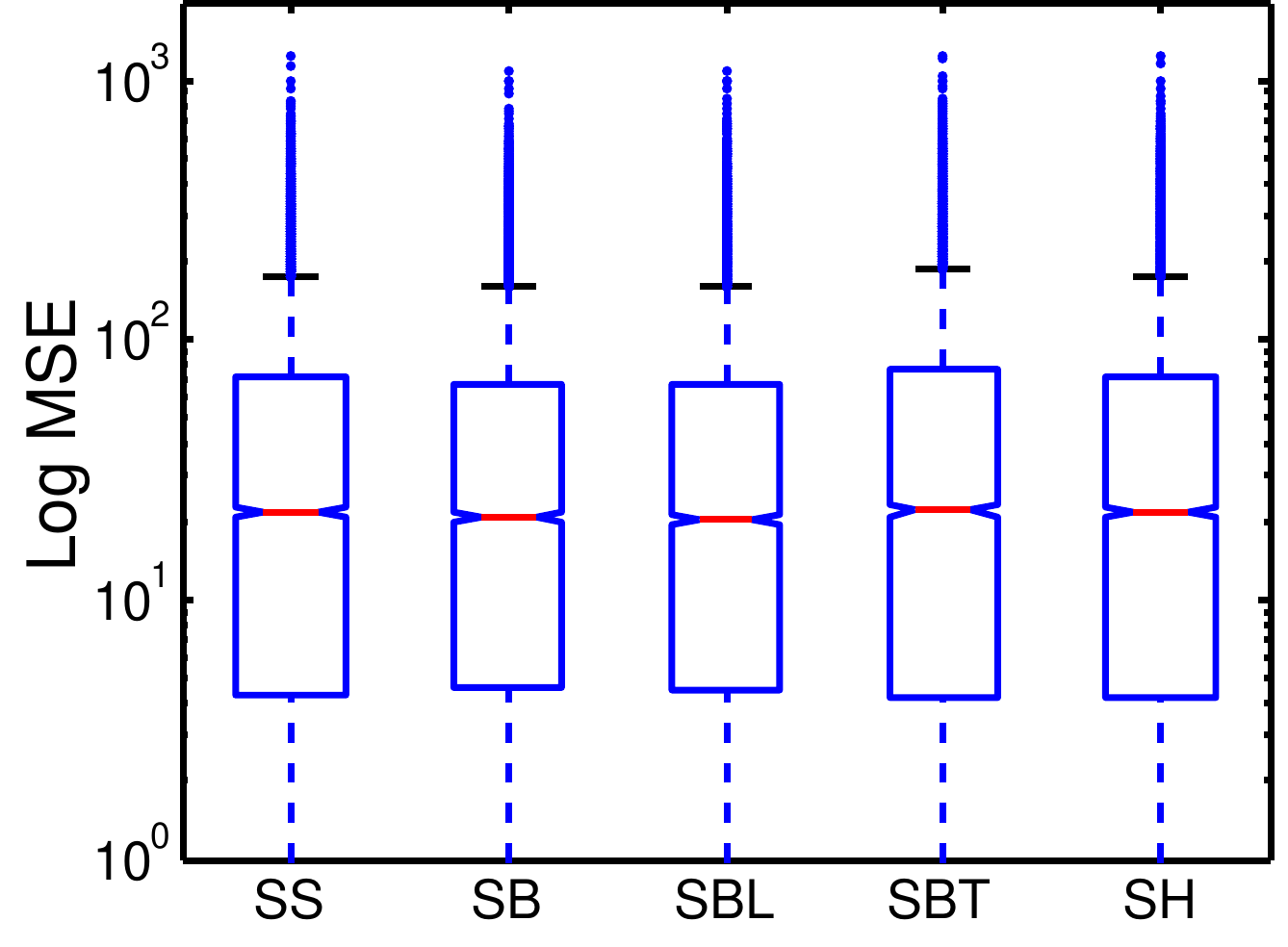}}
\caption{\label{fig:boxplot2}Empirical performance as measured by MSE, with $x$ drawn from a truncated Laplacian distribution and scaling coefficient $s$ fixed to $100$}
\end{center}
\end{figure*}

\subsection{Evaluation via Standard Wavelet Test Functions}

We next consider the standard set of univariate wavelet test functions: ``smooth,'' ``blocks,'' ``bumps,'' ``angles,'' ``spikes,'' and ``bursts,'' as illustrated in Fig.~\ref{fig:example1}.  A thorough comparative evaluation of several Poisson intensity estimation methods using these test functions is detailed in~\cite{ref:Besbeas_2004}, and here we repeat the same set of experiments using the estimators outlined above, along with the best-performing methods reviewed in~\cite{ref:Besbeas_2004}---including variance stabilization techniques currently in wide use as well as the more recent methods of~\cite{ref:Kolaczyk_1999b,ref:Timmermann_1999}.  To retain consistency with the experimental procedure of~\cite{ref:Besbeas_2004}, all methods except for~\cite{ref:Kolaczyk_1999b} were implemented using a 5-level translation-invariant wavelet decomposition;  the implementation of~\cite{ref:Kolaczyk_1999b} provided by~\cite{ref:Besbeas_2004} employs a decomposition level that is logarithmic in the data size, which we retained here.

As can be seen from Fig.~\ref{fig:example2}, the Skellam-based techniques we propose here measure
well against alternatives despite the diversity of features across these test functions, and the corresponding possibilities of model mismatch with respect to any assumed prior distribution of wavelet coefficients.  Overall, it can be seen that only the multiscale model of~\cite{ref:Kolaczyk_1999b} offers comparable performance.
\begin{figure*}
\centering
\subfigure[Smooth]{\includegraphics[width=.3\textwidth]{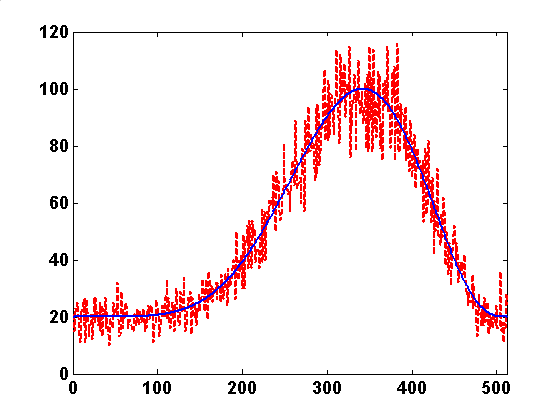}}
\subfigure[Blocks]{\includegraphics[width=.3\textwidth]{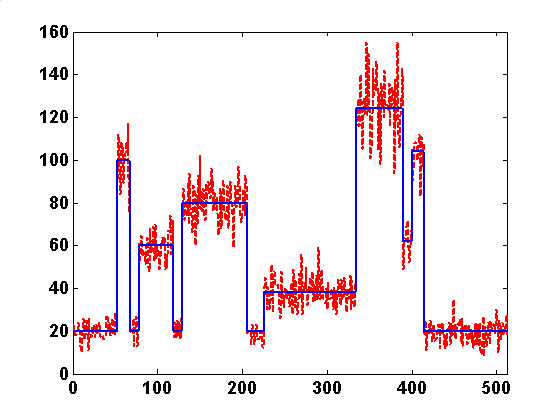}}
\subfigure[Bumps]{\includegraphics[width=.3\textwidth]{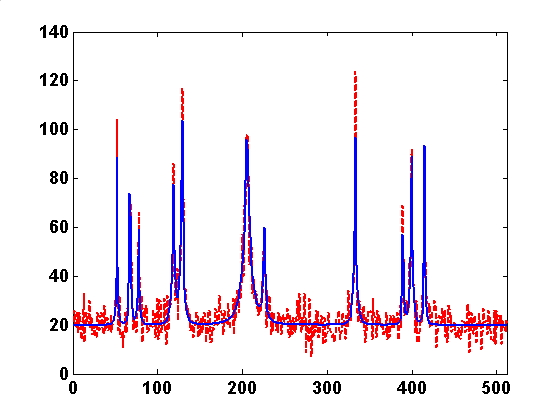}}
\subfigure[Angles]{\includegraphics[width=.3\textwidth]{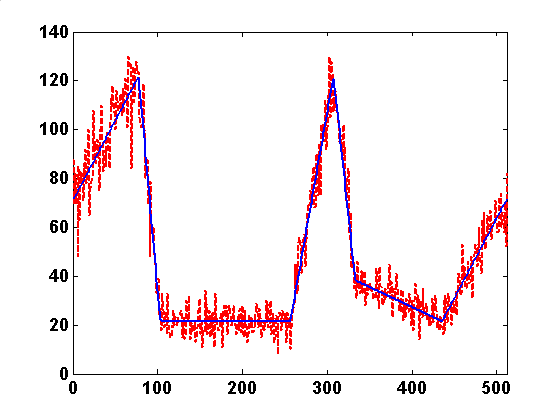}}
\subfigure[Spikes]{\includegraphics[width=.3\textwidth]{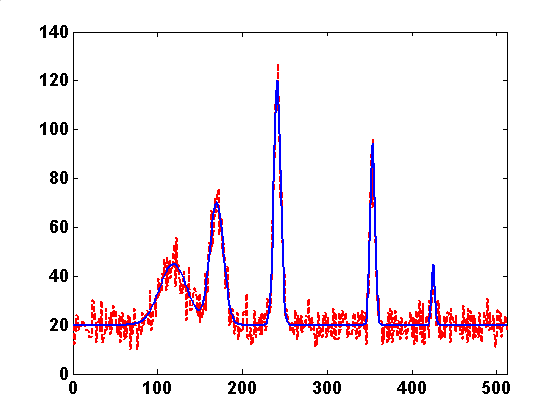}}
\subfigure[Bursts]{\includegraphics[width=.3\textwidth]{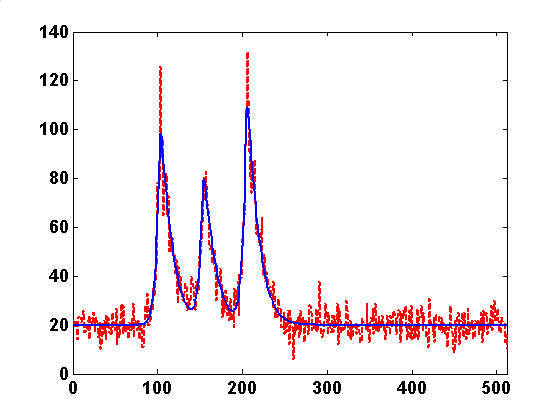}}
\caption{\label{fig:example1}Prototype intensity functions and corresponding Poisson-corrupted versions~\cite{ref:Besbeas_2004}}
\end{figure*}
\begin{figure*}
\centering
\subfigure[Smooth]{\includegraphics[width=.49\textwidth]{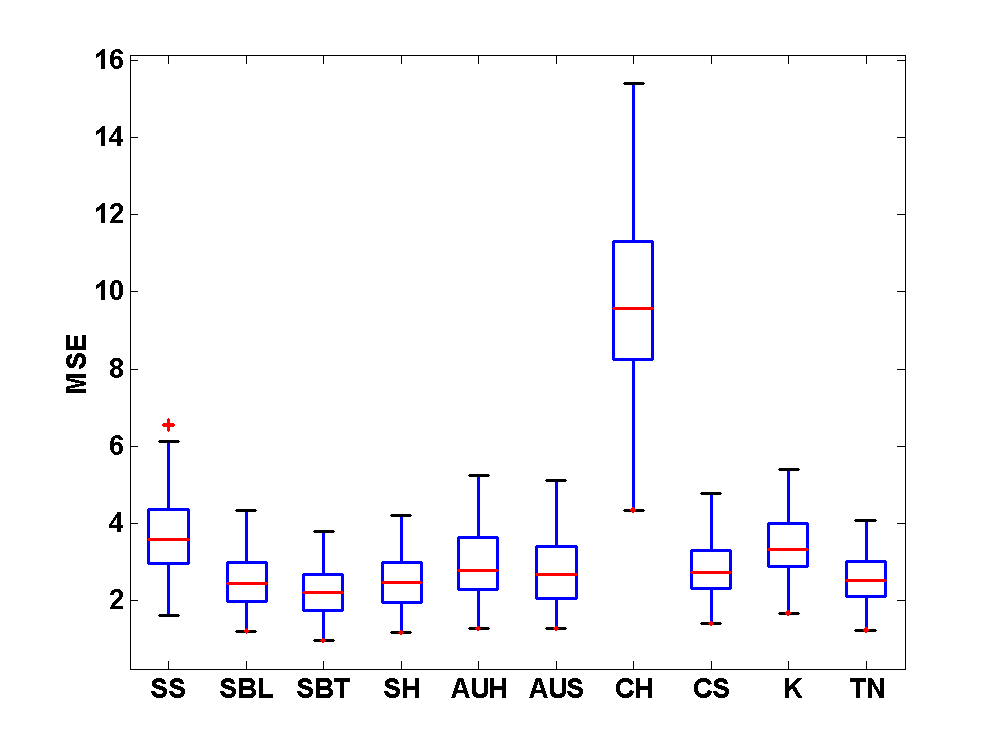}}
\subfigure[Blocks]{\includegraphics[width=.49\textwidth]{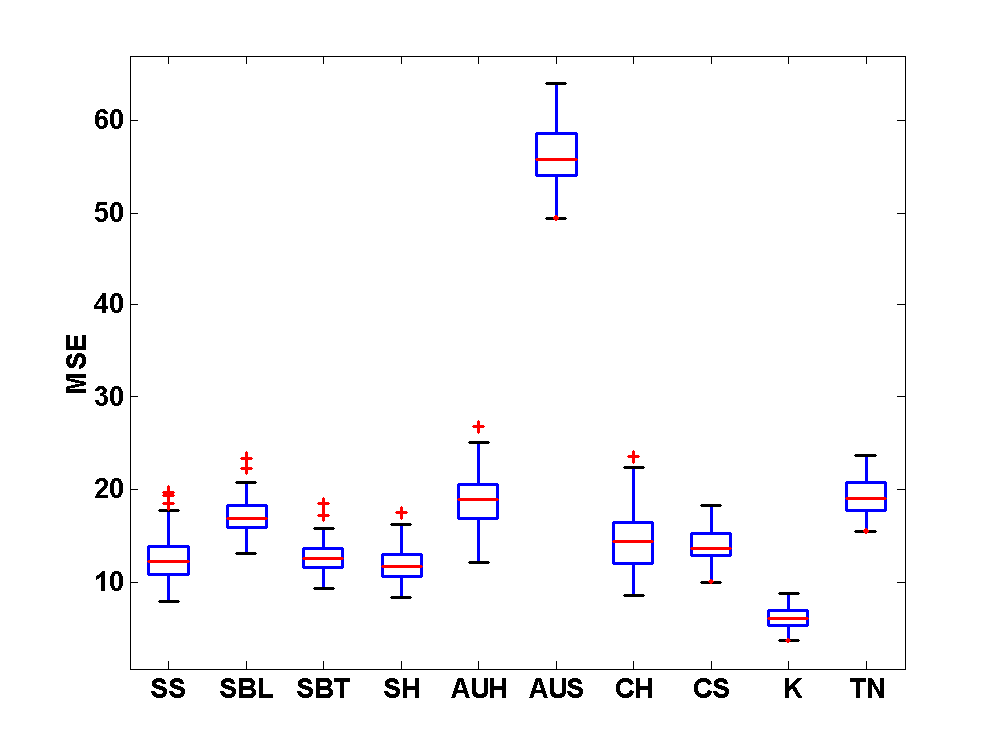}}
\subfigure[Bumps]{\includegraphics[width=.49\textwidth]{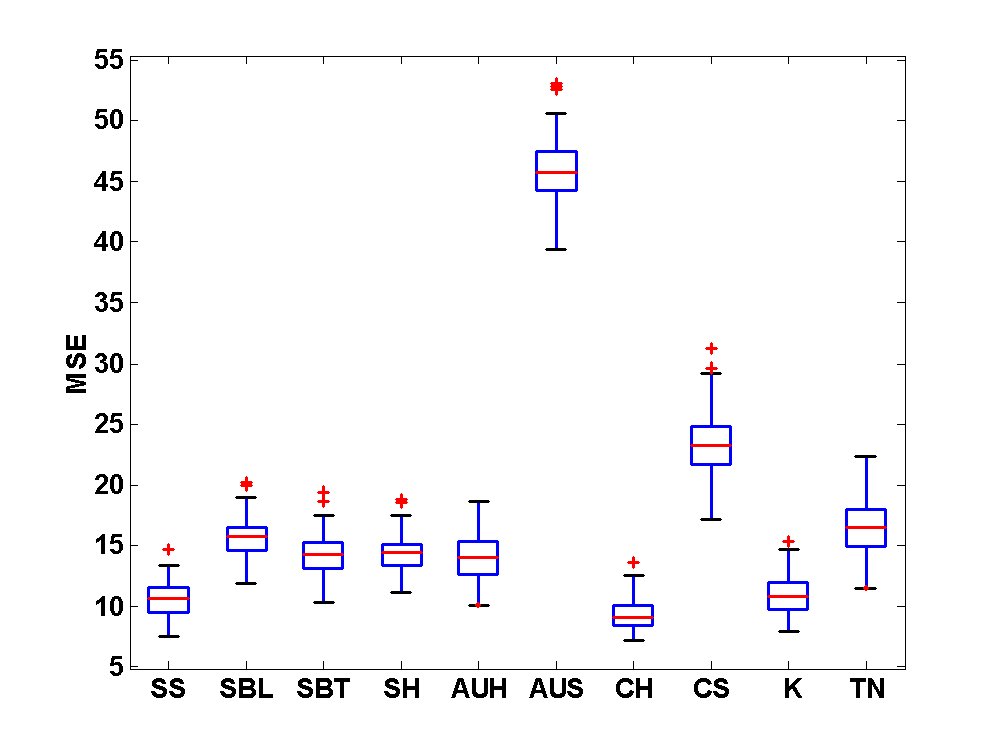}}
\subfigure[Angles]{\includegraphics[width=.49\textwidth]{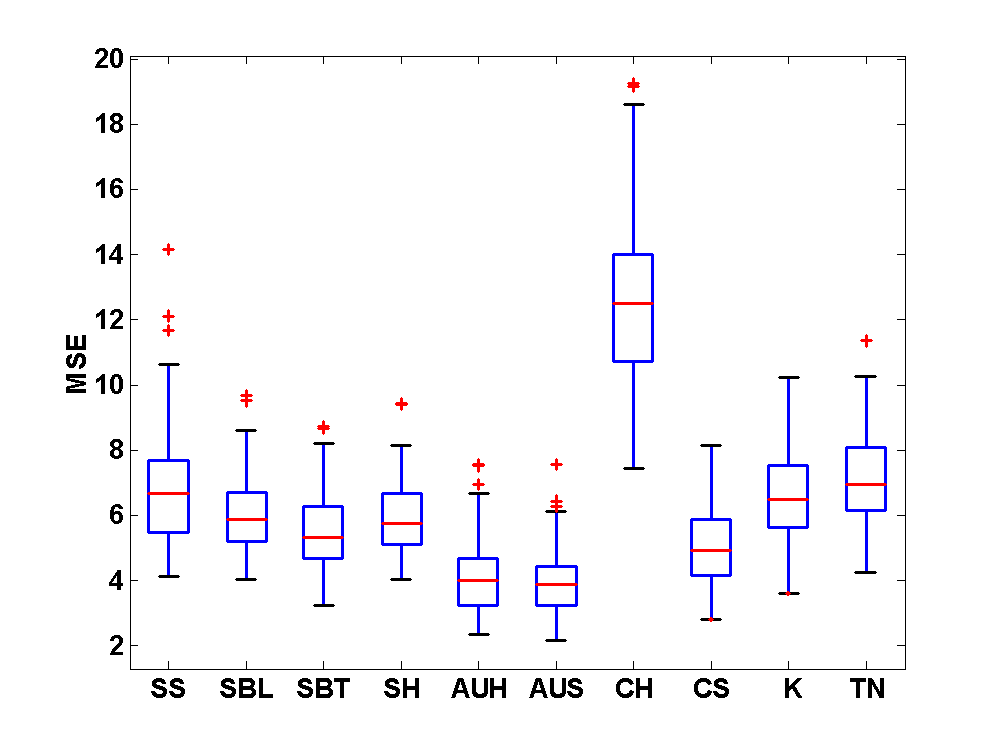}}
\subfigure[Spikes]{\includegraphics[width=.49\textwidth]{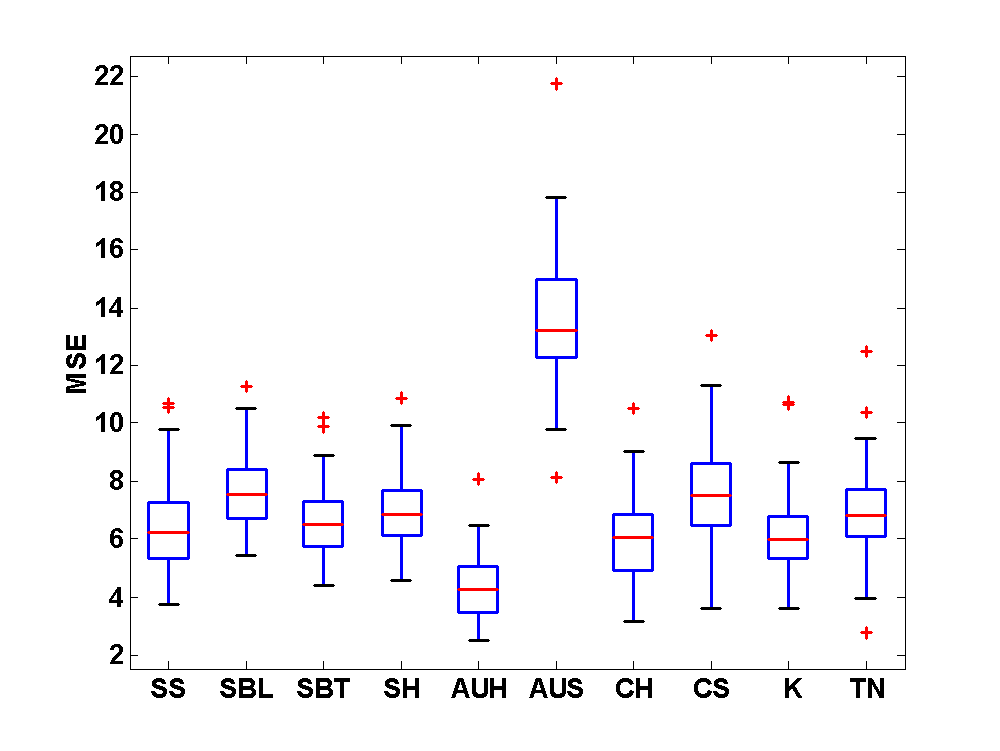}}
\subfigure[Bursts]{\includegraphics[width=.49\textwidth]{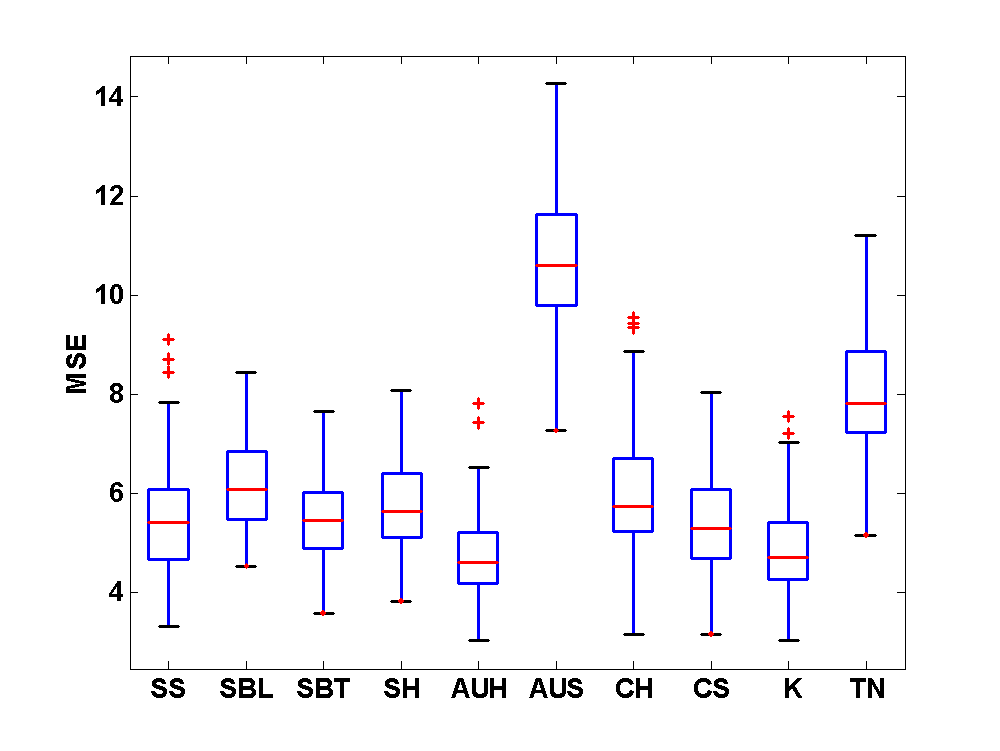}}
\caption{Mean-squared error, averaged over 100 trials, corresponding to reconstruction of the prototype
  functions of Fig.~\ref{fig:example1}.  (Note the difference in scale across the figure panels.) Skellam-based approaches comprise the left-hand portion of each figure panel; AUH/AUS denotes Anscombe variance stabilization \cite{ref:Anscombe_1948} with hard/soft universal thresholding~\cite{ref:Donoho_1994};
  CH/CS denotes corrected hard/soft thresholding~\cite{ref:Kolaczyk_1999a};
  K indicates the multiscale model of Kolaczyk~\cite{ref:Kolaczyk_1999b}; and
  TN is the multiscale multiplicative innovation model of Timmermann \& Nowak~\cite{ref:Timmermann_1999}
}
\label{fig:example2}
\end{figure*}

\subsection{Error and Perceptual Quality for Standard Test Images}

\begin{figure*}
\centering
\subfigure[Original test image]{\includegraphics[width=.49\textwidth]{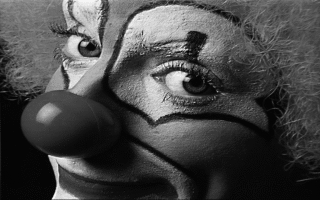}}
\subfigure[Noisy test image (SNR $\approx$ 10~dB)]{\includegraphics[width=.49\textwidth]{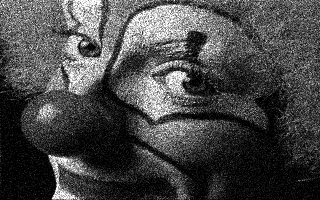}}
\subfigure[SkellamShrink (SS)]{\includegraphics[width=.49\textwidth]{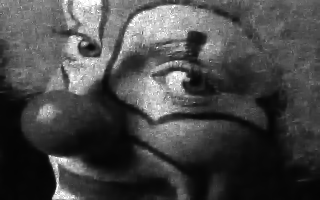}}
\subfigure[Bayes exact posterior mean (SB)]{\includegraphics[width=.49\textwidth]{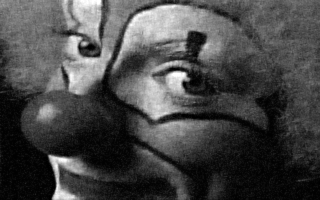}}
\subfigure[Bayes Laplacian thresholding (SBT)]{\includegraphics[width=.49\textwidth]{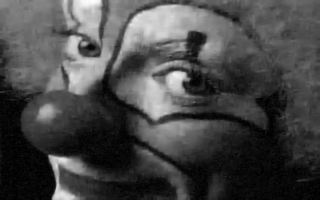}}
\subfigure[Adjusted-threshold hybrid (SH)]{\includegraphics[width=.49\textwidth]{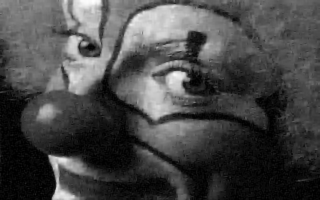}}
\subfigure[Haar-Fisz~\cite{ref:Fryzlewicz_2004} with~\cite{ref:Portilla_2003}]{\includegraphics[width=.49\textwidth]{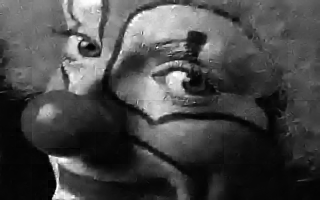}}
\subfigure[Multiscale multiplicative model~\cite{ref:Timmermann_1999}]{\includegraphics[width=.49\textwidth]{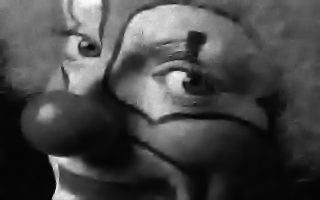}}
\caption{Performance comparison of the wavelet-based estimators derived in Section~\ref{sec:estimation} relative to existing approaches, shown for the ``clown'' test image}
\label{fig:example3}
\end{figure*}

\begin{table*}
\caption{Average reconstruction SNR~(dB) for a set of standard test images}
\label{tab:SNR}
\centering
\begin{tabular}{c|cccccccccc}
%{|@{~~}c@{~~}|@{~~}r@{/}r@{~~}r@{/}r@{~~}r@{/}r@{~~}r@{/}r@{~~}r@{/}r@{~~}r@{/}r@{~~}r@{/}r@{~~}r@{/}r@{~~}r@{/}r@{~~}r@{/}r@{~~}|}
&SNR&SS&SB&SBL&SBT&SH&\cite{ref:Anscombe_1948, ref:Donoho_1995}&\cite{ref:Fryzlewicz_2004, ref:Donoho_1995}&\cite{ref:Kolaczyk_1999b}&\cite{ref:Timmermann_1999}\\
%\multicolumn{2}{c@{~}}{noisy}&\multicolumn{2}{c@{~}}{SS}&\multicolumn{2}{c@{~}}{SB}&\multicolumn{2}{c@{~}}{SBL}&\multicolumn{2}{c@{~}}{SBT}&\multicolumn{2}{c@{~}}{SH}&\multicolumn{2}{c@{~}}{\cite{ref:Anscombe_1948, ref:Donoho_1995}}&\multicolumn{2}{c}{\cite{ref:Fryzlewicz_2004, ref:Donoho_1995}}&\multicolumn{2}{c}{\cite{ref:Kolaczyk_1999b}}&\multicolumn{2}{c|}{\cite{ref:Timmermann_1999}}\\
\hline\hline
Mean   & & 14.63 & 15.56 & 15.55 & 15.53 & 15.67 &  7.13 &  6.68
&  9.25 & 15.42 \\
Median &   & 14.28 & 15.24 & 15.26 & 15.31 & 15.35 &  6.35 &  6.67
&  9.30 & 15.30 \\
Min    & 0~dB  & 12.67 & 13.51 & 13.42 & 13.00 & 13.39 &  6.23 &  6.39
&  8.73 & 12.52 \\
Max    &  & 16.38 & 17.69 & 17.70 & 17.82 & 17.89 &  9.20 &  7.11
&  9.49 & 17.76 \\
Std.~Dev.&  &  1.38 &  1.62 &  1.64 &  1.78 &  1.70 &  1.23 &  0.30
&  0.25 &  1.90 \\
\hline
Mean &  & 15.96 & 16.71 & 16.69 & 16.64 & 16.85 & 10.36 &  9.79 &
10.99 & 16.39 \\
Median &  & 15.62 & 16.08 & 16.10 & 16.08 & 16.17 & 10.59 &  9.90
& 11.10 & 16.05 \\
Min& 3~dB & 14.28 & 14.73 & 14.68 & 14.27 & 14.80 &  8.86 &  8.97 &
10.25 & 13.15 \\
Max&  & 18.24 & 19.02 & 19.05 & 19.15 & 19.26 & 10.90 & 10.15 &
11.34 & 19.00 \\
Std.~Dev.&  &  1.53 &  1.75 &  1.78 &  1.92 &  1.84 &  0.71 &  0.42
&  0.35 &  2.15 \\
\hline
Mean& & 19.78 & 19.81 & 19.77 & 19.80 & 20.17 & 16.36 & 16.36 &
17.15 & 19.46 \\
Median& & 19.62 & 19.48 & 19.50 & 19.76 & 20.35 & 16.77 & 16.63 &
17.62 & 20.22 \\
Min& 10~dB & 17.58 & 17.68 & 17.51 & 17.46 & 17.71 & 15.22 & 15.41 &
15.82 & 16.36 \\
Max& & 22.27 & 22.23 & 22.22 & 22.34 & 22.61 & 16.87 & 17.15 &
18.20 & 22.20 \\
Std.~Dev.& &  1.89 &  1.86 &  1.89 &  2.01 &  1.98 &  0.65 &  0.66
&  0.95 &  2.37 \\
\end{tabular}
\end{table*}

\begin{table*}
\caption{Average reconstruction SSIM for images at 0, 3, and 10~dB SNR}
\label{tab:SSIM}
\centering
\begin{tabular}{c|cccccccccc}
%{|@{~~}c@{~~}|@{~~}r@{/}r@{~~}r@{/}r@{~~}r@{/}r@{~~}r@{/}r@{~~}r@{/}r@{~~}r@{/}r@{~~}r@{/}r@{~~}r@{/}r@{~~}r@{/}r@{~~}r@{/}r@{~~}|}
&Image&SS&SB&SBL&SBT&SH&\cite{ref:Anscombe_1948, ref:Donoho_1995}&\cite{ref:Fryzlewicz_2004, ref:Donoho_1995}&\cite{ref:Kolaczyk_1999b}&\cite{ref:Timmermann_1999}\\
%\multicolumn{2}{c@{~}}{noisy}&\multicolumn{2}{c@{~}}{SS}&\multicolumn{2}{c@{~}}{SB}&\multicolumn{2}{c@{~}}{SBL}&\multicolumn{2}{c@{~}}{SBT}&\multicolumn{2}{c@{~}}{SH}&\multicolumn{2}{c@{~}}{\cite{ref:Anscombe_1948, ref:Donoho_1995}}&\multicolumn{2}{c}{\cite{ref:Fryzlewicz_2004, ref:Donoho_1995}}&\multicolumn{2}{c}{\cite{ref:Kolaczyk_1999b}}&\multicolumn{2}{c|}{\cite{ref:Timmermann_1999}}\\
\hline\hline
Mean&0.050 & 0.463 & 0.523 & 0.523 & 0.538 & 0.547 & 0.172 & 0.132 &
0.214 & 0.524 \\
Median&0.048 & 0.439 & 0.537 & 0.540 & 0.570 & 0.572 & 0.140 & 0.118 &
0.200 & 0.563 \\
Min&0.026 & 0.398 & 0.471 & 0.474 & 0.396 & 0.461 & 0.084 & 0.076 &
0.124 & 0.323 \\
Max&0.083 & 0.571 & 0.581 & 0.584 & 0.613 & 0.614 & 0.318 & 0.247 &
0.404 & 0.604 \\
Std.~Dev.&0.024 & 0.057 & 0.044 & 0.047 & 0.078 & 0.058 & 0.086 & 0.063
& 0.096 & 0.099 \\
\hline
Mean&0.088 & 0.520 & 0.571 & 0.574 & 0.600 & 0.607 & 0.257 & 0.207 &
0.260 & 0.574 \\
Median&0.083 & 0.512 & 0.585 & 0.584 & 0.610 & 0.606 & 0.247 & 0.196 &
0.244 & 0.614 \\
Min&0.045 & 0.425 & 0.504 & 0.510 & 0.523 & 0.546 & 0.159 & 0.124 &
0.158 & 0.378 \\
Max&0.149 & 0.611 & 0.628 & 0.635 & 0.674 & 0.670 & 0.383 & 0.321 &
0.446 & 0.660 \\
Std.~Dev.&0.041 & 0.061 & 0.043 & 0.043 & 0.060 & 0.048 & 0.083 & 0.077
& 0.102 & 0.098 \\
\hline
Mean&0.261 & 0.698 & 0.683 & 0.693 & 0.731 & 0.731 & 0.474 & 0.465 &
0.516 & 0.707 \\
Median&0.228 & 0.694 & 0.676 & 0.692 & 0.752 & 0.748 & 0.459 & 0.440 &
0.501 & 0.707 \\
Min&0.151 & 0.628 & 0.593 & 0.602 & 0.646 & 0.651 & 0.335 & 0.341 &
0.404 & 0.645 \\
Max&0.431 & 0.786 & 0.791 & 0.791 & 0.789 & 0.807 & 0.641 & 0.634 &
0.666 & 0.778 \\
Std.~Dev.&0.108 & 0.053 & 0.063 & 0.060 & 0.056 & 0.056 & 0.121 & 0.116
& 0.103 & 0.055 \\
\end{tabular}
\end{table*}

We now consider an image reconstruction scenario using a test set of well-known 8-bit gray scale test images that feature frequently in the engineering literature: ``Barbara,'' ``boat,'' ``clown,'' ``fingerprint,'' ``house,'' ``Lena,'' and ``peppers.''  Corresponding pixel values are considered as the true underlying intensity function of interest; both noise level characterization and reconstruction results are reported in terms of signal-to-noise ratio (SNR) in decibels, a quantity proportional to log-MSE.  By way of competing approaches we consider~\cite{ref:Timmermann_1999,ref:Kolaczyk_1999b,ref:Donoho_1995,ref:Portilla_2003}, with~\cite{ref:Donoho_1995,ref:Portilla_2003} used in conjunction with the variance stabilization methods of~\cite{ref:Anscombe_1948,ref:Fryzlewicz_2004}.  Implementations were set at an equal baseline implementation comprising a 3-level undecimated Haar wavelet decomposition, with no a priori neighborhood structures assumed amongst the coefficients.

The performance of the Skellam methods proposed here offers noticeable improvements over alternative approaches, in terms of visual quality (Fig.~\ref{fig:example3}), mean-squared error (Table~\ref{tab:SNR}), and perceptual error (Table~\ref{tab:SSIM}).  In terms of visual quality, we have generally observed that the proposed Skellam Bayes approaches yield restored images in which the spatial smoothing is appropriately locally adaptive---for example, these methods yield effective noise attenuation in both bright (see forehead) and dark (see
black background) regions of the example image shown in Fig.~\ref{fig:example3}.  A comparison of Figs.~\ref{fig:example3}(c) and~(f) reveals the importance of incorporating the scaling coefficient $s$ explicitly in the estimator; images processed via SBT tended to be similar to those for which SH was used,  but with softer edges.  In comparison, methods based on variance stabilization typically fail to completely resolve the heteroscedasticity of the underlying process, as evidenced by the under- and over-smoothed noise in bright regions such as the forehead and hair textures of Fig.~\ref{fig:example3}(g). The Bayesian method of~\cite{ref:Timmermann_1999} typically yields far smoother output images, in which texture information is almost entirely lost; see, for example, the hair in Fig.~\ref{fig:example3}(h). With the exception of SS, Skellam-based estimation methods suffer considerably less from the reconstruction artifacts typically associated with wavelet-based denoising, as can be seen in the cheek structure of the ``clown'' image.

We also report numerical evaluations of estimator performance in this setting, by way of both SNR in Table~\ref{tab:SNR} and the widely-used perceptual error metric of Structural Similarity Index (SSIM)~\cite{ref:Wang_2004} in Table~\ref{tab:SSIM}, for input SNR of 0, 3, and 10~dB.  The results readily confirm that Skellam-based approaches outperform competing alternatives, with only that of~\cite{ref:Timmermann_1999} remaining competitive---though as described above, its oversmoothing results in a great deal of loss of texture.  The SkellamShrink adjusted-threshold hybrid (SH) method measures the best in terms of both SNR and SSIM, with other Skellam-based approaches generally outperforming all alternatives save for~\cite{ref:Timmermann_1999}.

\section{Discussion}
\label{sec:discussion}

In this article we derived new techniques for wavelet-based Poisson intensity estimation by way of the Skellam distribution.  Two main theorems, one showing the near-optimality of Bayesian shrinkage and the other providing for a means of frequentist unbiased risk estimation, served to yield new estimators in the Haar transform domain, along with low-complexity algorithms for inference.  A simulation study using standard wavelet test functions as well as test images confirms that our approaches offer appealing alternatives to
existing methods in the literature---and indeed subsume existing variance-stabilization approaches such as Haar-Fisz by yielding exact unbiased risk estimates---along with a substantial improvement for the case of enhancing image data degraded by Poisson variability.  We expect further improvements for specific applications in which correlation structure can be assumed a priori amongst Haar coefficients, in a manner similar to the gains reported by~\cite{ref:Portilla_2003} for the case of image reconstruction in the presence of additive noise.

\bibliographystyle{acmtrans-ims}

\end{document}